%% file: arxiv.tex
\documentclass[letterpaper, 11pt]{article}
\input{settings/packages}

\input{settings/macros}

\title{The Economic Security of Exponential EIP-1559}

\author{%
\begin{tabular}{ccc}
Ben Berger & Edward W. Felten & Robin Fritsch \\[0.2em]
\texttt{bberger@offchain.io} &
\texttt{ed@offchain.io} &
\texttt{rfritsch@offchain.io} \\[1em]
\multicolumn{3}{c}{Offchain Labs}
\end{tabular}%
}

\date{}

\newtheorem{lemma}{Lemma}
\newtheorem{corollary}{Corollary}
\newtheorem{observation}{Observation}
\newtheorem{claim}{Claim}
\newtheorem{proposition}{Proposition}

\theoremstyle{definition} 
\newtheorem{example}{Example}
\newtheorem{definition}{Definition}
\declaretheorem[name=Theorem, numberwithin=section]{theorem}

\begin{document}

\maketitle

\begin{abstract}

We consider the problem of parameter selection for EIP-1559, a widely used gas pricing mechanism for blockchains.
As opposed to previous work, we aim to achieve \emph{economic security}, where the chosen parameters guarantee a lower bound on total collected gas fees whenever usage over a given time window exceeds a specified threshold. This bound can be set prohibitively high, thereby economically deterring usage above a desired level.
Such guarantees are particularly relevant to long-term objectives such as limiting state growth.

We apply our approach to the pure exponential version of EIP-1559 which Ethereum uses and to a variant used by  Robinhood Chain and Arbitrum.

To achieve economic security, we first characterize the revenue-minimizing gas-usage distribution for the exponential version and for the other variant we construct a distribution whose associated fee revenue provably approximates the minimum.

Using these results, we then demonstrate how to set economically secure mechanism parameters for both variants and we discuss the associated tradeoffs.

\end{abstract}

\section{Introduction}
\label{sec:introduction}

EIP-1559 \cite{eip1559} refers to the protocol upgrade that introduced Ethereum's current gas fee mechanism.\footnote{Gas is an abstract unit that measures blockchain compute resources in EVM-compatible chains.}
Since its activation in the London hard fork in August 2021, EIP-1559 has been extensively studied by the academic blockchain community. Variants of its fee mechanism have been adopted across numerous blockchain systems, particularly EVM-compatible chains and Layer 2 networks, including Robinhood Chain, Arbitrum, Base, Optimism, Monad, and Avalanche.

Although deployed variants differ in various ways, they all share several key characteristics.
First, the \emph{basefee} --- the price of a single gas unit --- is updated from one block to the next as a deterministic function of the preceding block's gas usage.
Each mechanism specifies a \emph{gas target} (or more than one target in the case of Robinhood Chain and Arbitrum), typically chosen to reflect the chain's sustainable throughput.
The basefee for the next block increases when the current block's gas usage exceeds the target and decreases when usage falls below it.
This feedback mechanism is intended to drive gas usage toward the target, with the basefee converging toward a market-clearing price for gas. 

A second shared characteristic is that the basefee responds approximately exponentially to deviations of gas usage from the target, subject to a minimum-basefee floor.
More precisely, the multiplicative change in the basefee is approximately exponential in the deviation of gas usage from the target.
The deployed variants in practice differ slightly in how they approximate this exponential response.\footnote{For example, Ethereum's basefee update rule multiplies the current price by the linear approximation to $\exp(d/(8T))$, where $d$ is the deviation from the target $T$ in the current block.
On the other hand, Ethereum's blob gas basefee update rule approximates the exponential by a Taylor representation \cite{eip4844}.}
% On the other hand, Ethereum's blob gas basefee update rule approximates a Taylor representation, adding extra rounded-down terms as long as the latest term evaluates to at least 1 \cite{eip4844}.}

A third key feature is that the target lies well below the protocol-imposed block gas limit.
This gap allows gas usage to exceed the sustainable throughput for some time, before the resulting increase in the basefee pushes the gas usage back towards the target.

This raises the question of how to set the mechanism's \emph{adjustment rate}, namely the magnitude of the relative price response to a given deviation from the target.
Intuitively, the greater the cost or risk associated with such deviations, the more aggressively the basefee should respond.
At the same time, users generally benefit from smoother and more predictable price dynamics.\footnote{\label{fn:wallets_max_price_increase}
For example, wallets rely on bounds on the rate at which the basefee can increase when estimating transaction fees. If the basefee can rise too quickly, a fee quote presented when a transaction is submitted may no longer be sufficient by the time the transaction is included on-chain.}
Choosing the adjustment rate therefore entails a trade-off between rapidly correcting deviations from the target and maintaining stable, predictable prices.

EIP-1559-style pricing mechanisms can enforce throughput constraints arising from different system objectives and operating over different time scales.

One important short-horizon constraint is the processing capacity of the chain's infrastructure such as its nodes.
When gas usage exceeds sustainable throughput, nodes with weaker hardware or bandwidth fall behind the tip of the chain, potentially by many seconds. Users relying on such nodes then experience longer transaction confirmation times and act on a stale view of the chain state, increasing the likelihood of transaction failures. Both effects severely degrade the user experience on the chain.
This motivates relatively aggressive price adjustments when gas usage approaches or exceeds the maximum sustainable node throughput, even over short periods, so that such episodes are corrected quickly.

Longer-horizon throughput constraints can impose different requirements on an EIP-1559-style pricing mechanism.
A notable example is controlling state growth, where the goal is to limit the growth of the disk footprint of the chain state over the course of longer time periods, such as a month or a year \cite{ethfoundation2025futurestate}. Here, deviating from the implied per-block target over the course of a few blocks is largely inconsequential.
Accordingly, such deviations can be met with significantly more gradual price adjustments between consecutive blocks.
In particular, the mechanism designer has greater flexibility when choosing the adjustment rate for a long-term constraint, making its calibration more subtle than in the short-term setting.

Moreover, the mechanism designer can adjust not only the adjustment rate but also the target itself.
When the objective is to bound average gas usage over a given time window, setting the target below the desired long-run usage bound creates additional slack, allowing the constraint to be satisfied with a less aggressive adjustment rate.

Although EIP-1559 variants that account for both short- and long-term considerations have been deployed or are pending deployment --- for example, Arbitrum uses a gas target to control state growth \cite{offchain2025aipgas}, while Ethereum plans to meter ``state-growth'' gas separately \cite{eip8037} --- we are not aware of prior work that systematically relates these considerations to the choice of mechanism parameters.
%Although EIP-1559 variants that account for both short- and long-term considerations have been deployed or are pending deployment \cite{offchain2025aipgas,eip8037}, we are not aware of prior work that systematically informs the choice of mechanism parameters based on these considerations.

Our goal is to provide a principled method for setting mechanism parameters --- both adjustment rates and targets --- grounded by a theoretical foundation. 
We hope this framework can inform the design of future EIP-1559-style mechanisms as well as the calibration of existing systems.

\subsection{Our Contributions}

In this work we develop an approach for setting the parameters of EIP-1559 style pricing systems, aimed to achieve the following \emph{economic security guarantee}:
\begin{quote}
    If the gas usage rate exceeds some pre-defined safety threshold over a given time window, then the pricing system must have collected a prohibitively high amount of fees during that time window. Assuming the latter cannot happen, the conclusion is that the safety threshold cannot be violated.
\end{quote}
We achieve this through two conceptual steps, to be described below at a high level.

\paragraph{Finding the Cheapest Distributions of a Fixed Gas Amount}
%First, given an instantiation of the pricing mechanism with specified target and adjustment rate values among other parameters, we identify the cheapest way --- or an approximation thereof --- to distribute gas over a given time window $n$ which results in the \emph{minimum} attainable fee revenue, under the constraint that the average gas usage over the given time window is at least some given amount $U$ --- presumably, $U$ is the safety threshold mentioned above.
Fix an instantiation of the pricing mechanism, including its target, adjustment rate, and other parameters.
Given a time window of length $n$ and a usage threshold $U$, we identify the cheapest way---or an approximation thereof---to distribute gas usage across the window while maintaining an average usage of at least $U$.
Equivalently, we minimize the total fee revenue generated over the window subject to this usage constraint.
We hereafter refer to such a fee minimizing distribution as an \emph{optimal} distribution.

The set of optimal distributions depends heavily on the set parameters of the pricing system.
We prove the existence of an optimal (or approximately optimal) strategy of a particular structure; we call such strategies \emph{simply structured}.

%To find an optimal distribution, which depends heavily on the set parameters of the pricing system, we fix an arbitrary gas distribution $G$ with average gas usage $U$ over the given time window. 
%We then perform a sequence of transformations on $G$ where each transformation can only decrease the fee revenue relative to the previous one (or increase by an upper-bounded amount), while maintaining the same average usage $U$.  

%In particular, we show that as long as the current gas distribution does not admit some simple structure, then such a transformation is always possible.  Since each transformation provably results in a distribution that has not been considered in the sequence before, this procedure must end with a distribution that admits said structure.
%Furthermore, since $G$ was arbitrary, it might as well have been an optimal distribution; we conclude that there exists an optimal (or approximately optimal) distribution admitting the simple structure.

The class of simply-structured distributions is sufficiently restricted that the optimal strategy among them can be found efficiently, and we describe how one might perform that search.
Together with the structural result above, this yields an optimal, or approximately optimal, strategy overall.
%Having established the previous step, we know that the obtained strategy is optimal (or approximately optimal) overall.

\paragraph{Determining Economically Secure Parameters}
%In this step, we assume that we are given a \emph{safe fee revenue amount} $B$, which is presumably the amount of funds that the mechanism designer believes could never be collected in fees over a time window of length $n$, since the amount is too high.
In this step, we assume a given \emph{safe fee revenue amount} $B$, representing a level of fee revenue that the mechanism designer considers infeasible to be spent over a time window of length $n$.

%Then, for each target out of a predetermined discrete interval of candidates (which lies entirely below the safety threshold $U$), we search for the least aggressive adjustment rate for which the minimal fee revenue is at least $B$.
Then, for each target in a predetermined discrete set of candidates, all lying below the safety threshold $U$, we determine the least aggressive adjustment rate for which the minimum attainable fee revenue is at least $B$.

%This process results in a curve on the plane consisting of (target, adjustment rate) points, each corresponding to a choice of parameters that is $B$-economically-secure.
This process results in a curve in the (target, adjustment rate) plane, where each point corresponds to a parameter choice that is $B$-economically secure.

\paragraph{Analyzed Variants}

% We note that there could be other objectives to setting the pricing parameters (e.g. optimizing for revenue or other business purposes) which might interfere with the economic security objective.

We analyze two variants of the EIP-1559 pricing mechanism.
The first is the pure exponential version of Ethereum's original variant \cite{feist2022exponential}.
In this setting, we are able to precisely pinpoint the optimal gas distribution for any given set of parameters.

We then study a variant of EIP-1559 with two gas targets. This mechanism is deployed on Robinhood Chain, while a slightly more complex version of it is deployed on Arbitrum.\footnote{It is likely that Arbitrum will adopt the variant studied in this paper in the future.}
In this case, the optimal gas distribution can be highly complex for certain parameter configuration. We are however able to find a simply-structured distribution whose corresponding fee revenue approximates that of the optimal distribution up to an approximation factor which depends on the set parameters.  

Notably, this approximation factor is shown to be at least twice as good as the highest possible relative price increase in a single block transition, denoted $\tau$. 
This observation is useful since in practice there are external factors that force the mechanism designer to upper bound $\tau$.\footref{fn:wallets_max_price_increase}
%\footnote{For example, wallets rely on bounded price increase rates for UX purposes:  if the price can increase too quickly, then the quote presented by the wallet to a transaction sender might turn out to be too low to cover the transaction's fees at the time of inclusion on-chain.}

Arbitrum and Robinhood Chain, for example, impose a maximum relative price increase of 20\% per second \cite{offchain2025aipgas}.
For any similar mechanism with the same maximum price increase, our results prove a 10\% approximation to the true minimum --- and for Robinhood Chain’s specific parameters, a 6.3\% approximation.
%This in turn implies that for any practical choice of Robinhood Chains's pricing parameters, the minimal fee revenue amount reported by our algorithm constitutes a 10\% approximation to the true minimum.
%\bbc{Go over this.  If the max increase is $r$, then formally the approximation factor is $1+\frac{r}{2}$}
    
\subsection{Structure of the Paper}
Section~\ref{sec:preliminaries} introduces the EIP-1559 variants discussed above, along with useful notation.  Hereafter, the pure exponential version of EIP-1559 is referred to as the ``single target mechanism'', and Robinhood Chain's mechanism is referred to as the ``two target mechanism''.

Section~\ref{sec:single_target} proves the existence of a simply-structured gas distribution that achieves minimum fee revenue for the single target mechanism (Theorem~\ref{thm:structure}), while Section~\ref{sec:two_target_mechanism} establishes an approximately optimal simply-structured distribution for the two target mechanism (Theorem~\ref{thm:approximation}).

%In Section~\ref{sec:single_target} we analyze the single target mechanism, demonstrating the existence of a simply-structured gas distribution that achieves minimum fee revenue --- see Theorem~\ref{thm:structure}.

%In Section~\ref{sec:two_target_mechanism} we analyze the two target mechanism, demonstrating the existence of an approximately optimal simply-structured distribution --- see Theorem~\ref{thm:approximation}.

In Section~\ref{sec:computing_optimal_revenue} we compute the revenue of the optimal simply-structured gas distribution in both mechanisms.
In Section~\ref{sec:setting_parameters} we demonstrate how one might set the parameters in practice and discuss the tradeoffs that arise for that decision.

\subsection{Related Work}
EIP-1559-style mechanisms have been studied extensively in recent years from a variety of perspectives.
\cite{roughgarden2020transactionfeemechanismdesign_eip1559} and \cite{roughgarden2021transactionfeemechanismdesign} provide rigorous treatments of the incentive guarantees of Ethereum's mechanism, whereas \cite{diamandis2023designing} connects EIP-1559-style price updates to gradient-based convex optimization, with the adjustment rate playing the role of a step size.

Several papers have studied how the mechanism's adjustment rate impacts its performance.
\cite{leonardos2021dynamical} and its empirical counterpart \cite{reijsbergen2021honeymoon} use dynamical methods to study the relation between the adjustment rate and the resulting gas usage dynamics. They identify conditions under which sufficiently small adjustment rates lead to convergence of the gas usage rate to the target, while larger adjustment rates can lead to instability and chaos.
\cite{zhivkov2026dynamic} evaluates different adjustment rate values under simulated demand scenarios, examining metrics such as gas utilization and basefee volatility.

Both \cite{ndiaye2026blockchain} and \cite{crapis2025optimal} relate the optimal adjustment rate of EIP-1559-style mechanisms to the gas demand elasticity, with the objective of stabilizing gas usage around its target.
\cite{azouvi2023base} analyzes a base-fee manipulation attack whose expected profitability for proposers increases with the adjustment rate.
\cite{BabaioffNisan2025Welfare} demonstrates that with patient bidders, EIP-1559 achieves approximately optimal welfare, where the obtained approximation factor is closely related to the adjustment rate.

Importantly, these papers assume a gas demand elasticity, whereas we take an adversarial approach resulting in worst-case security bounds.

\section{Preliminaries}
\label{sec:preliminaries}

We consider a finite time horizon of $n \in \mathbb{N}_{>0}$ rounds.
In each round \(i \in \{1,\ldots,n\}\), the integer-valued gas usage is denoted by
\(g_i\), and is subject to the given per-round gas limit $L \in \mathbb{N}_{>0}$, so that
$0 \leq g_i \leq L$.

A \emph{strategy} is a gas-usage vector
$
G=(g_i)_{i=1}^n.
$ 
Fix an average gas usage \(U \in [0,L]\), and assume that \(nU\) is an
integer. We say that \(G\) is a \emph{feasible} strategy if
\[
\sum_{i=1}^n g_i=nU.
\]
We denote the set of feasible strategies by
\[
\mathcal{G}
:=
\left\{
G\in\{0,\ldots,L\}^n:
\sum_{i=1}^n g_i=nU
\right\}.
\]
Given a strategy \(G\),
a gas pricing mechanism
%the mechanism 
determines a \emph{price
trajectory}
\[
P(G)=(p_i)_{i=1}^n,
\]
where \(p_i\) is the price charged for gas used in round \(i\).
%, and the gas usage \(g_i\) determines the price in the following round.
%
The \emph{revenue} associated with \(G\) is
\[
\operatorname{Rev}(G)
:=
\sum_{i=1}^n p_i g_i.
\]
Our objective is to find a revenue minimizing feasible strategy.
A strategy \(G^*\) is \emph{optimal} if
\[
\operatorname{Rev}(G^*)
=
\min_{G\in\mathcal{G}} \operatorname{Rev}(G)
\]
For \(z\in\mathbb{R}\), we use the notation
\[
[z]_+:=\max\{z,0\}.
\]

\subsection{The Single Target Mechanism}
\label{sec:single-target-preliminaries}

The single-target mechanism is specified by an integer valued \emph{gas target} $0<T<L$, an \emph{adjustment rate} $A>0$, a minimum gas price $p_{\min} \in \mathbb{R}_{>0}$ and a starting price $p_{\mathrm{start}} \geq p_{\min}$.
Given the price \(p_i\) and gas usage \(g_i\) in round \(i\), the price in
the following round is
\[
p_{i+1}
=
\max\left\{
p_{\min},
\;
p_i
\exp\left(
\frac{g_i-T}{AT}
\right)
\right\},
\qquad
i=1,\ldots,n.
\]
Although \(p_{n+1}\) is determined by the final gas usage \(g_n\), it does
not enter the revenue over the \(n\)-round horizon.

% \rfc{I find the following confusing because it's not clear why it's needed (all prices are already defined) and what a "gas backlog" is. Alternative below.}

% We define the corresponding \emph{gas backlog}
% \[
% z_i:=AT\log\left(\frac{p_i}{p_{\min}}\right).
% \]
% Then
% \[
% z_1
% =
% AT\log\left(\frac{p_{\mathrm{start}}}{p_{\min}}\right) \ \ , \ \ 
% z_{i+1}
% =
% [z_i+g_i-T]_+ \ \ , \ \ \  i = 1,\ldots,n.
% \]
% The corresponding price is recovered from
% \[
% p_i
% =
% p_{\min}
% \exp\left(\frac{z_i}{AT}\right).
% \]

Alternatively, the single-target mechanism can be defined by using a \emph{gas backlog} $z_i$ to track how much more gas than the target has been used up to round $i$.
The backlog is maintained by
\[
z_{i+1} = [z_i+g_i-T]_+ \ \ , \ \ \  i = 1,\ldots,n,
\]
and the gas price is computed as
\[
p_i = p_{\min} \exp\left(\frac{z_i}{AT}\right).
\]
The initial backlog $z_1$ that corresponds to a starting price of $p_1 = p_{start}$ is
\[
z_1 = AT\log\left(\frac{p_{\mathrm{start}}}{p_{\min}}\right).
\]
We assume that $p_{\mathrm{start}}$ is defined so that $z_1\in \mathbb{N}_{\geq 0}$.
We observe that $p_i = p_{\min}$ iff $z_i=0$.

\subsection{The Two Target Mechanism}
\label{sec:two-target-preliminaries}

The two-target mechanism is specified by two integer valued gas targets
\[
0<T_\ell<T_h<L,
\]
two corresponding adjustment rates
\[
A_\ell>A_h>0,
\]
and a minimum price $p_{\min}\in \mathbb{R}_{>0}$.
The mechanism maintains two integer valued gas backlogs:
\[
\ell_i\geq 0 \ \ , \ \ 
h_i\geq 0
\]
with initial values
\[
\ell_1=\ell_{\mathrm{start}}\ \ , \ \ 
h_1=h_{\mathrm{start}}
\]
which are also specified by the mechanism.
We interpret \(\ell_i\) as the backlog associated with the lower target
\(T_\ell\), and \(h_i\) as the backlog associated with the higher target
\(T_h\).

Given gas usage \(g_i\), the two backlogs evolve according to
\[
\ell_{i+1}
=
[\ell_i+g_i-T_\ell]_+\ \ , \ \ 
h_{i+1}
=
[h_i+g_i-T_h]_+,
\]
The price in round \(i\) is
\[
p_i
=
p_{\min}
\exp\left(
\frac{\ell_i}{A_\ell T_\ell}
+
\frac{h_i}{A_hT_h}
\right).
\]
and the starting price is given by $p_{\mathrm{start}}=p_1$.
In particular, if both initial backlogs equal 0 then $p_{\mathrm{start}}=p_{\min}$.

Note that the gas price can increase by at most a factor of
\[
    1+\tau := \exp\left( \frac{L-T_\ell}{A_\ell T_\ell} + \frac{L-T_h}{A_hT_h}\right)
\]
from one round to the next, which occurs when $L$ units of gas are consumed during the round.

Given the mechanism parameters, the initial state $\ell_1, h_1$, and a strategy \(G\),
the state trajectories
\[
(\ell_i)_{i=1}^{n} \ \ , \ \ 
(h_i)_{i=1}^{n},
\]
and hence the price trajectory
\[
P(G)=(p_i)_{i=1}^n,
\]
are uniquely determined.

\begin{example}
\label{running_example}
A practical example for the mechanism parameters,  similar to those recently used on Robinhood Chain, is:
\begin{align*}
    L        &= 192~\mathrm{Mgas/s},    \quad p_{\min} = 0.02~\mathrm{gwei}\\
    T_\ell   &= 10~\mathrm{Mgas/s},     \quad A_\ell = 86{,}400,\\
    T_h      &= 60~\mathrm{Mgas/s},     \quad A_h = 12.
\end{align*}
This corresponds to a maximum price increase per round of $\tau = 20\%$.

Moreover, a realistic example of a constraint is $U = 11~\mathrm{Mgas/s}$ over $n= 2{,}592{,}000~\mathrm{s}$ (30 days).
We will use these parameters as a \emph{running example} throughout the paper.
\end{example}

\section{The Single Target Mechanism}
\label{sec:single_target}

In this section we identify an optimal strategy for the single target mechanism that has a simple structure.  
Let us fix the time horizon $n$, average gas usage $U$, gas limit $L$ and starting price $p_{\mathrm{start}}$ as described in Section~\ref{sec:preliminaries}.  
We also fix a specific single target mechanism with parameters as described in Section~\ref{sec:single-target-preliminaries}:  a target $T$, adjustment rate $A$ and minimum price $p_{\min}$.
As discussed in that section, a strategy $G=(g_i)$ determines the price trajectory through
\[
p_{i+1}
=
\max\left\{
p_{\min},
\;
p_i
\exp\left(
\frac{g_i-T}{AT}
\right)
\right\},
\qquad
i=1,\ldots,n.
\]
Throughout this section we omit the denominator $\frac{\cdot}{AT}$ from the exponent terms to lighten the notation (equivalently, ``$\exp$'' is interpreted as the base number $e^{\frac{1}{AT}}$).

The rest of this section is devoted to proving

%\begin{theorem}
\begin{restatable}{theorem}{thmstructure}
\label{thm:structure}
There is an optimal strategy $G$ satisfying:
\begin{enumerate}
    \item $G$ is monotone non-decreasing, i.e. $g_i \leq g_{i+1}$ for all $i$.
    \item There is at most one round $i$ for which $0 < g_i < T$. 
    \item There is at most one round $i$ for which $T < g_i < L$.
    % \item If such $j$ exists, then $p_{j+1}=p_{\min}$. \bbc{Prove the following:  Let $j$ be the first round in which $g_j \geq T$.  Then $j$ is also the first round in which $p_j = p_{\min}$, \textbf{or} we must have $g_j = g_{j+1} = \cdots = g_n = L$}
    \item Let $i$ be the first round for which $g_i \geq T$.  Then $p_i = p_{\min}$, or $g_{i+1} = \cdots = g_n = L$.
    \item
    Let $i$ be the first round for which $p_i = p_{\min}$.  
    %If there is a round $j < k < n$ such that $T < g_k$ then $g_j \geq T$.
    Then $g_i \geq T$, or
    there is at most a single round $k$ such that $g_k > T$.
    %all rounds $k$ satisfy $g_k \leq T$.
    % \item \bbe{
    % Suppose $p_{\min} < p_{\text{start}}$, let $1 < j$ be the first round for which $T \leq g_j$, and suppose that $j$ is also the first round for which $p_j = p_{\min}$.
    % }
\end{enumerate}
\end{restatable}
%\end{theorem}
%
Informally, the theorem says that, save for degenerate cases where $U$ is extremely high or extremely low, there is an optimal strategy that behaves as follows --- in the first rounds it uses no gas until the price decreases all the way to the minimum;  at that point it raises the gas usage to exactly the target with at most a single transition round $0<g_i<T$.  The strategy sustains this for as long as possible, keeping the price at the minimum, until the total gas usage constraint (given by $U$) forces it to raise the usage to the limit through final round, again with at most a single transition round $T<g_i<L$.

%To this end, we consider a total order relation between feasible strategies: ordering them lexicographically.
%We say the strategy $G = (g_i)$ \emph{is dominated} by the strategy $G' = (g'_i)$, denoted by $G' \succ G$, when $G'$ is lexicographically smaller than $G$ (see the appendix for a formal definition).

To this end, we first define a total order relation between feasible strategies, essentially ordering them by lexicographic ordering.
We say the strategy $G = (g_i)$ \emph{is dominated} by the strategy $G' = (g'_i)$, denoted by $G' \succ G$, if:
\begin{itemize} 
\item $[g'_1 < g_1]$, or 
\item $[g'_1 = g_1 \ \& \ g'_2 < g_2]$, or
\end{itemize}

$\vdots$
\begin{itemize}
\item $[g'_1 = g_1 \ \& \ \cdots \ \& \ g'_i = g_i \ \& \ g'_{i+1} < g_{i+1}]$, or
\end{itemize}

$\vdots$
\begin{itemize}
\item $[g'_1 = g_1 \ \& \ \cdots \ \& \ g'_{n-1} = g_{n-1} \ \& \ g'_n < g_n]$.
\end{itemize}

Note that the last condition will never hold as that would contradict the assumption that both $G, G'$ are $U$-average.

We say that a strategy $G$ is \emph{$j$-shiftable} for some round $1 \leq j < n$ if $g_j > 0$ and $g_{j+1} < L$.  In this case it is possible to transfer some gas units from round $j$ to $j+1$.

We say that an ordered pair of strategies $(G,G')$ is \emph{$j$-shifted} for some round $1\leq j < n$ if $G'$ is obtained from $G$ by transferring some gas units from round $j$ to round $j+1$.  Formally:
\begin{enumerate}
    \item For all $i\notin \{j,j+1\}$, $g_i=g'_i$.
    \item $g'_j < g_j$
    \item $g_j+g_{j+1} = g'_j+g'_{j+1}$.
\end{enumerate}
Note that in this case we have $G' \succ G$.

% \rfc{
% If I see correctly, this definition is used in the proof of Theorem 1: Once when applying Proposition 1, and once when proving Bullet 4. I suggest:
% (1) using $G'$ when applying Proposition 1 (to be consistent and since the form of $G^j$ is not relevant then)
% (2) when proving Bullet 4, defining the appropriate $G'$ in place, and removing this definition
% }
% \bbc{done}

% For a $j$-shiftable strategy $G$ we define $G^j$ as the strategy obtained from $G$ by transferring exactly one gas unit from round $j$ to $j+1$, i.e.:
% \begin{align*}
% g^j_i &= g_i \text{ for all } i \notin \{j, j+1\} \\
% g^j_j &= g_j - 1 \\
% g^j_{j+1} &= g_{j+1} + 1 
% \end{align*}
% %Note that $G^j$ is also a $U$-average strategy and that it satisfies $G \succ G^j$.
% We denote the price trajectory corresponding to $G^j$ by $P^j = (p^j_i)_{i=1}^n$.
% Note that $(G,G^j)$ is a $j$-shifted pair.

\begin{lemma}
\label{lem:rev G - rev G' cancel out}
    Let $(G,G')$ be a $j$-shifted strategy pair. 
    %Suppose that one of the following holds:
    % \begin{enumerate}
    %     \item $g_j \geq T$ and $g'_j \geq T$
    %     \item $g_{j+1} \leq T$ and $g'_{j+1} \leq T$
    % \end{enumerate}
    % {\color{cyan}[slightly more general version] 
    Suppose the following two conditions hold:
    \begin{enumerate}
        \item $g_j \geq T$ or $g_{j+1} \leq T$
        \item $g'_j \geq T$ or $g'_{j+1} \leq T$
    \end{enumerate}
    Then $p_{j+2} = p'_{j+2}$. In particular, the difference in revenue between $G$ and $G'$ is determined only by rounds $j$ and $j+1$, i.e.
    \[
    \mathrm{Rev}(G) - \mathrm{Rev}(G') = 
    p_j g_j + p_{j+1} g_{j+1} - \left(p_j g'_j + p'_{j+1} g'_{j+1}\right)
    \]
\end{lemma}

\begin{proof}
    First note that  $p_j=p_j'$ since $G$ and $G'$ are identical before round $j$.
    Inserting the definition of $p_{j+1}$ into that of $p_{j+2}$ yields
    \begin{equation}\label{eq:p_j+2_max3}
        p_{j+2} = \max \left \{ p_j \exp\left(g_j + g_{j+1} -2T\right), p_{min} \exp\left( g_{j+1} - T\right), p_{min} \right\}.
    \end{equation}
    Consider the two cases of Condition 1 of the lemma.
    If the first, $g_j \geq T$, holds, then this together with $p_j\geq p_{min}$ implies that the first expression in \eqref{eq:p_j+2_max3} is larger or equal to the second. Hence,
    \begin{equation}\label{eq:p_j+2_max2}
        p_{j+2} = \max \left \{ p_j \exp\left(g_j + g_{j+1} -2T\right), p_{min} \right\}.
    \end{equation}
    On the other hand, the second case, $g_{j+1} \leq T$, implies $p_{min} \exp(g_{j+1} - T) < p_{min}$ which also shows that \eqref{eq:p_j+2_max2} holds.
    Hence, Condition 1 guarantees that \eqref{eq:p_j+2_max2} holds. 
    
    The analog expression to \eqref{eq:p_j+2_max2} for $p_{j+2}'$ follows from Condition 2. Finally, $g_j+g_{j+1} = g'_j+g'_{j+1}$ and $p_j=p_j'$ imply $p_{j+2} = p'_{j+2}$, as required.
    
\end{proof}

\begin{proposition}\label{prop:g_j_leq_T_g_j+1+leq_T-1}
Let $G$ be a 
%$j$-shiftable 
strategy such that
$g_j \leq T$ and $g_{j+1} \leq T-1$, for some round $j$.
Consider the strategy $G'$ defined as $G$ only that $g'_j = g_j - 1$, $g'_{j+1} = g_{j+1} +1$, and note that $(G,G')$ is a $j$-shifted strategy pair. Then
%If $g_j \leq T$ and $g_{j+1} \leq T-1$, then
$\mathrm{Rev}(G') \leq \mathrm{Rev}(G)$.
\end{proposition}

\begin{proof}
% We first compare the price trajectories of $G$ and $G^j$ at $i=j+1$ and $i=j+2$.
% For $i=j+2$, one of the expressions over which the max is taken in Claim \ref{clm:aux} is different in $P$ as opposed to $P^j$ --- in one we have $p_{\min} \cdot \exp(g_{j+1}-T)$ and in the other $p_{\min} \cdot \exp(g^j_{j+1}-T)$, but these are dominated by $p_{\min}$ since $g_{j+1}, g^j_{j+1} \leq T$.
% Thus we have $p_{j+2} = p^j_{j+2}$.
% For $i=j+1$, we have $p^j_{j+1} \leq p_{j+1} \leq p_j$, where the first inequality holds since $g^j_j < g_j$, and the second inequality holds since $g_j \leq T$.
% Combining with Lemma \ref{lem:p_i=p'_i}, we have
Note that both conditions of Lemma~\ref{lem:rev G - rev G' cancel out} hold for the $j$-shifted pair $(G,G')$.  Thus we have

\begin{align*}
\mathrm{Rev}(G) - \mathrm{Rev}(G') &= p_j g_j + p_{j+1} g_{j+1} - \left(p_j g'_j + p'_{j+1} g'_{j+1}\right)\\ 
&= p_j g_j + p_{j+1} g_{j+1} - \left( p_j(g_j-1) + p'_{j+1}(g_{j+1}+1) \right) \\
&= (p_j - p'_{j+1})  + g_{j+1}(p_{j+1}-p'_{j+1}) 
\end{align*}
and this is non-negative since $p'_{j+1} \leq p_{j+1} \leq p_j$.

\end{proof}

% \begin{proposition}
% \label{prop: T < g_j+1 < g_j}
%     Let $G$ be a strategy for which $T < g_{j+1} < g_j$ for some round $j$.  Consider the strategy $G'$ defined as $G$ only that $g'_j = g_{j+1}$, $g_{j+1} = g_j$.  Notice that $G' \succ G$.  Then $\mathrm{Rev} (G') < \mathrm{Rev}(G)$.
% \end{proposition}

% \begin{proof}
% In this case the assumptions of Lemma \ref{lem:p_j+2=p'_j+2} are satisfied and in particular we have $p_{j+2} = p'_{j+2}$.  Combining with Lemma \ref{lem:p_i=p'_i}, we have \bbc{It remains to justify the Taylor approximation below.}
% \begin{align*}
%     \mathrm{Rev}(G) - \mathrm{Rev}(G') &= p_jg_j +p_{j+1}g_{j+1} - (p_j g_{j+1} + p'_{j+1}g_j) \\
%     &= p_j\left[
%     g_j + \exp(g_j - T)g_{j+1} - g_{j+1}-\exp(g_{j+1}-T)g_j 
%     \right] \\
%     &\approx p_j\left[
%     g_j + (1 + g_j - T)
%     g_{j+1} - g_{j+1}-
%     (1 + g_{j+1}-T)g_j 
%     \right] \\
%     &\approx p_j
%     T (g_j - g_{j+1}) \\
%     & > 0
% \end{align*}
% \end{proof}

\begin{proposition}\label{prop:g_j + g_j+1 - T leq L}
Let $G$ be a strategy such that $g_j > T$ and
%$g_{j+1} = 0$.
$g_j+g_{j+1}-T \leq L$, for some round $j$.
Consider the strategy $G'$ defined as $G$ only that 
$g'_j = T$, 
%$g'_{j+1} = g_j - T$.
$g'_{j+1} = g_{j+1}+ g_j - T$,
and note that 
%$G' \succ G$.
$(G,G')$ is a $j$-shifted strategy pair.
Then $\mathrm{Rev}(G') \leq \mathrm{Rev}(G)$.
\end{proposition}

\begin{proof}
Note that in this case we have $p'_{j+1} = p_j$, $p_{j+1} > p_j$, 
% and that the assumptions of Lemma \ref{lem:p_i=p'_i} and Lemma \ref{lem:p_j+2=p'_j+2} are again satisfied.
and that both conditions of Lemma~\ref{lem:rev G - rev G' cancel out} hold.
Therefore we have

\begin{align*}
\mathrm{Rev}(G) - \mathrm{Rev}(G') &= p_j g_j +p_{j+1}g_{j+1} - (p_j T + p_j(g_{j+1} + g_j -T)) \\
&= g_{j+1}(p_{j+1} - p_j) \\
% &= g_{j+1}p_j(\exp(g_j - T) - 1) \\ 
&\geq 0
\end{align*}

\end{proof}

\begin{proposition}
\label{prop:g_j + g_j+1 - T > L}
Let $G$ be a strategy such that $g_j > T$, $g_{j+1} < L$ and
%$g_{j+1} = 0$.
$g_j+g_{j+1}-T > L$, for some round $j$.
Consider the strategy $G'$ defined as $G$ only that 
$g'_j = g_j + g_{j+1} - L$, 
%$g'_{j+1} = g_j - T$.
$g'_{j+1} = L$,
and note that $(G,G')$ is a $j$-shifted strategy pair.
Then $\mathrm{Rev}(G') \leq \mathrm{Rev}(G)$.
\end{proposition}

% \begin{proof}
%     Note that the assumptions of Lemma \ref{lem:p_i=p'_i} and Lemma \ref{lem:p_j+2=p'_j+2} are satisfied.  Therefore we have \bbc{It remains to justify the Taylor approximation below}
%     \begin{align*}
%         \mathrm{Rev}(G) - \mathrm{Rev}(G') &= p_j g_j +p_{j+1}g_{j+1} - p_j(g_j + g_{j+1} - L) - p'_{j+1}L \\
%         &=p_j g_j +p_j \exp(g_j - T)g_{j+1} - p_j(g_j + g_{j+1} - L) - p_j L \exp(g_j+g_{j+1} - L - T) \\
%         &= p_j\left[
%          \exp(g_j - T)g_{j+1} - g_{j+1} + L - L\exp(g_j+g_{j+1} - L - T)
%         \right] \\
%         &\approx p_j\left[
%          (1 + g_j - T)g_{j+1} - g_{j+1} + L - L(1+g_j+g_{j+1} - L - T)
%         \right] \\
%         &=p_j\left[
%          ( g_j - T)(g_{j+1} - L) - L(g_{j+1} - L)
%         \right] \\ 
%         &= p_j\left[
%          ( g_j - T - L)(g_{j+1} - L)
%         \right] \\
%         &>0
%     \end{align*}
% \end{proof}

\begin{proof}
    % The hypotheses give $g_j > T$, hence $p_{j+1} = p_j \exp(g_j - T)$, and
    % $g'_j = g_j + g_{j+1} - L > T$ (using $g_j + g_{j+1} - T > L$), hence
    % $p'_{j+1} = p_j \exp(g'_j - T) = p_j \exp(g_j + g_{j+1} - L - T)$.
    % Thus the assumptions of Lemma~\ref{lem:p_i=p'_i} and Lemma~\ref{lem:p_j+2=p'_j+2}
    % hold, so only rounds $j$ and $j+1$ contribute to the revenue difference:
    Note that both conditions of Lemma~\ref{lem:rev G - rev G' cancel out}
    hold.  Therefore we have
    
    \begin{align*}
        \mathrm{Rev}(G) - \mathrm{Rev}(G')
        &= p_j g_j + p_{j+1} g_{j+1} - p_j(g_j + g_{j+1} - L) - p'_{j+1} L \\
        &= p_j g_j + p_j \exp(g_j - T)\, g_{j+1} \\
        &\qquad - p_j(g_j + g_{j+1} - L) - p_j L \exp(g_j + g_{j+1} - L - T) \\
        &= p_j \Bigl[\, g_{j+1}\bigl(\exp(g_j - T) - 1\bigr)
                       - L\bigl(\exp(g_j + g_{j+1} - L - T) - 1\bigr) \,\Bigr].
    \end{align*}
    Set $a := g_j - T$; in particular we have $0 < a$ and $0 <  L - a$. Define
    \[
        F(x) \;:=\; x\bigl(\exp(a) - 1\bigr) - L\bigl(\exp(x + a - L) - 1\bigr),
        \qquad x \in [L - a,\, L].
    \]
    The bracketed expression above equals $F(g_{j+1})$, and the proposition assumptions imply that
    $g_{j+1} \in (L - a,\, L)$.
    % : the lower bound rewrites
    % $g_j + g_{j+1} - T > L$, and the upper bound is $g_{j+1} < L$.
    It therefore
    suffices to show that $F(x) \geq 0$ on $[L - a,\, L]$.

    Differentiating twice gives
    $F''(x) = -L\, \exp(x + a - L) < 0$,
    so $F$ is strictly concave on $[L - a,\, L]$.  It suffices therefore to show that $F$ is non-negative at the end points.  Indeed, we have
    \begin{align*}
        F(L) &= L\bigl(\exp(a) - 1\bigr) - L\bigl(\exp(a) - 1\bigr) = 0, \\
        F(L - a) &= (L - a)\bigl(\exp(a) - 1\bigr) - L\bigl(\exp(0) - 1\bigr)
                  = (L - a)\bigl(\exp(a) - 1\bigr) \;\geq\; 0.
    \end{align*}
    % For any $x \in [L - a,\, L]$, writing $x$ as the convex combination
    % $x = \tfrac{L - x}{a}(L - a) + \tfrac{x - (L - a)}{a}\, L$, concavity yields
    % \[
    %     F(x) \;\geq\; \tfrac{L - x}{a}\, F(L - a) + \tfrac{x - (L - a)}{a}\, F(L)
    %           \;=\; \tfrac{(L - x)(L - a)\bigl(\exp(a) - 1\bigr)}{a}
    %           \;\geq\; 0.
    % \]
    Hence $\mathrm{Rev}(G) - \mathrm{Rev}(G') = p_j\, F(g_{j+1}) \geq 0$, as required.
    
\end{proof}

% \rfc{When reading the proof it's a bit annoying that the bullet points are 3 pages up. maybe restate them here?}
We are now ready to proof Theorem~\ref{thm:structure}.  For clarity we restate it here.

\thmstructure*

\begin{proof}[Proof of Theorem~\ref{thm:structure}]
% Let $G$ be an optimal strategy. 
% It is enough to show that if any of the bullets does not hold but the ones preceding it do hold, then we can find a strategy $G'$ 
% %that satisfies that bullet and also dominates $G$.
% such that either $(G,G')$ is a $j$-shifted strategy pair and $G'$ is optimal (bullets 1-4), or $G'$ is strictly revenue-decreasing in contradiction to the optimality of $G$ (bullet 5).
% This is indeed enough since the number of strategies is finite and repeating the above in iterations always obtains a new strategy that we has not been considered before (since $\succ$ is an order relation).
% Therefore the process must end with an optimal strategy which satisfies all bullets.
% \rfc{I find this paragraph very difficult to understand! I suggest separating the 1-4 and 5 logic (see below).}
%
Let $G$ be an optimal strategy.
To prove bullets 1-4 it suffices to show that if any of the bullets does not hold but the ones preceding it do hold, then we can find an optimal strategy $G'$ 
such that $(G,G')$ is a $j$-shifted strategy pair.
This is indeed enough since the number of strategies is finite and repeating the above in iterations always obtains a new strategy that we has not been considered before (since $\succ$ is an order relation).
Therefore the process must end with an optimal strategy which satisfies all bullets.

Finally, we prove bullet 5 by showing that bullets 1-4 holding while bullet 5 does not, contradicts the optimality of $G$.

To this end, let us assume that bullet 1 does not hold, i.e., there is some round $j$ for which $g_j > g_{j+1}$. If $g_j \leq T$, we can apply Proposition \ref{prop:g_j_leq_T_g_j+1+leq_T-1} to $G$.
The obtained strategy $G'$ satisfies $\mathrm{Rev}(G') \leq \mathrm{Rev}(G)$ and therefore $G'$ must be optimal as well.
Similarly, if $g_j > T$ and 
%$g_{j+1} \geq 1$ then we can apply Theorem \ref{thm:2}. 
%$g_{j+1} \leq T$, 
$g_j+g_{j+1}-T \leq L$,
then we can apply Proposition \ref{prop:g_j + g_j+1 - T leq L}.
If $g_j > T$
%and $g_{j+1} = 0$, then we apply Theorem \ref{thm:3}. 
%and $g_{j+1} > T$, then we can apply Theorem \ref{prop: T < g_j+1 < g_j}.
and $g_j+g_{j+1}-T > L$, then we can apply Proposition \ref{prop:g_j + g_j+1 - T > L} (note that $g_{j+1} < g_j$ implies in particular that $g_{j+1} < L$).
{This completes the treatment of bullet 1.} 

Now let us assume that $G$ satisfies bullet 1 but does not satisfy bullet 2. The monotonicity of $G$ implies the existence of two consecutive rounds $j, j+1$ such that $0 < g_j, g_{j+1} < T$.
But note that in this case 
%the assumptions of Proposition \ref{prop:g_j_leq_T_g_j+1+leq_T-1} hold and so $G^j$ is also optimal and dominates $G$.
we can apply Proposition~\ref{prop:g_j_leq_T_g_j+1+leq_T-1} to obtain the desired strategy $G'$.

Bullet 3 can be proven similarly to bullet 2, but using Proposition \ref{prop:g_j + g_j+1 - T leq L} or Proposition \ref{prop:g_j + g_j+1 - T > L} instead of Proposition \ref{prop:g_j_leq_T_g_j+1+leq_T-1}, depending on whether we have $g_j+g_{j+1}-T \leq L$ or not.

Assume bullets 1,2 and 3 hold but bullet 4 does not hold.  Let $i$ be the first round for which $g_i \geq T$. Since bullet 4 does not hold and $G$ is monotone non-decreasing then we have $p_j > p_{\min}$ and $g_j \leq g_{j+1} < L$.  First, note that if $T < g_j$ then we can use either Proposition \ref{prop:g_j + g_j+1 - T leq L} or \ref{prop:g_j + g_j+1 - T > L} to find a strategy $G'$ such that $(G,G')$ is a $j$-shifted pair and $\mathrm{Rev}(G') \leq \mathrm{Rev}(G)$ as we want.

We therefore assume that $g_j = T$.  
Consider the strategy $G'$ defined as $G$ only that $g'_j = g_j - 1$, $g'_{j+1} = g_{j+1} +1$, and note that $(G,G')$ is a $j$-shifted strategy pair.
To prove $\mathrm{Rev}(G') \leq \mathrm{Rev}(G)$, we first argue that $p_{j+2} = p'_{j+2}$ (note that as opposed to previous cases, we cannot directly apply Lemma~\ref{lem:rev G - rev G' cancel out} here).
To this end, we first observe that  $p_j' = p_j > p_{\min}$ implies $p_j' \geq p_{\min}\exp(1)$, since prices are always of the form $p_{\min}\exp(z)$ for $z \in \mathbb{N}_{\geq 0 }$.  This in turn implies
\[
p_j' \exp (g'_j-T) = p_j' \exp (g_j-1-T) = p_j' \exp(-1) \geq p_{\min},
\] 
which further implies $p'_{j+1} = \max\{p'_{j}\exp(g'_{} - T), p_{\min}\} =  p_j \exp (g'_j-T)$ and
\begin{align*}
p'_{j+2}   &= \max\{p'_{j+1}\exp(g'_{j+1} - T), p_{\min}\} \\ 
            &= \max\{p_j'\exp(g'_j + g'_{j+1} - 2T), p_{\min}\}.
\end{align*}
The equivalent identity for $p_{j+2}$ follows from $p_j \exp(g_j - T) = p_j \geq p_{\min}$.
Finally, $g_j+g_{j+1} = g'_j+g'_{j+1}$ implies $p_{j+2} = p'_{j+2}$.
From this point, the proof that $\mathrm{Rev}(G') \leq \mathrm{Rev}$ is identical to the proof of Proposition~\ref{prop:g_j_leq_T_g_j+1+leq_T-1}.

Now suppose bullets 1,2,3,4 hold but bullet 5 does not, and let $j$ be the first round such that $p_j = p_{\min}$.  By assumption, we have $g_j < T$, and there are at least two rounds after $j$ with usage strictly above $T$.
In particular, there is a round $j<k\leq n - 1$ such that $g_k > T$, 
which in turn implies by the monotonicity of $G$ that $p_{\min} < p_n$ and $ T < g_n$.  Consider the strategy $G'$ obtained from $G$ by transferring 1 gas unit from round $n$ to round $j$.  The price trajectory stays exactly the same, implying that the only change in revenue is from the price paid for that transferred gas unit.  From here we get $\mathrm{Rev}(G') = \mathrm{Rev}(G)-(p_n -p_{\min}) < \mathrm{Rev}(G)$, contradicting the optimality of $G$.

\end{proof}

The following observation is a byproduct of the proof and will be used later in the paper.

\begin{observation}
\label{obs:ending_price_can_only_decrease}
    All strategy transformations applied in the proof of Theorem~\ref{thm:structure} do not increase the ending price $p_{n+1}$.
\end{observation}

In fact, only the transformation proving bullet 5 might strictly decrease the price, while all other transformations in the theorem's proofs leave the ending price unchanged.

\begin{corollary}\label{cor:final}
    Suppose that $p_{\text{start}} = p_{\min}$
    and that $U \geq T$.
    %and that $Un \geq (T-1)n + L + 1$ \bbc{$Un \geq UT$}.  
    Suppose also that $(U-T)n  = (L-T)q + r$, for some $q,r \in \mathbb{N}_{\geq 0}$ such that $0\leq r < L-T$. 
    % Then there is an optimal strategy $G$ such that
    % \[g_i = 
    % \begin{cases} 
    %     T   & \text{if } i< n - q \\
    %     T + r   & \text{if } i = n-q \\
    %     L   & \text{if } i > n-q
    % \end{cases}
    % \]
    Then there is an optimal strategy $G$ of the form
    \[
    G = \big(\underbrace{T,\ldots,T}_{n-q-1\text{ rounds}},\;\; T+r,\;\; \underbrace{L,\ldots,L}_{q\text{ rounds}}\big)
    \]
\end{corollary}

\begin{proof}
    By Theorem \ref{thm:structure}, there is an optimal strategy $G$ that satisfies all bullets in the theorem statement.  Notice that by our assumption, the first round $j$ for which $p_j = p_{\min}$ is $j=1$.  We split to cases:
    \paragraph{$\bm{Un - T(n-1) \leq L }$:}  In this case, the strategy $G$ with $g_1=\cdots=g_{n-1}=T$ and $g_n = Un - T(n-1)$ is feasible.
    It is optimal since it attains $p_{\min}$ in all rounds --- and it is exactly the strategy stated in the corollary.
    
    \paragraph{$\bm{Un - T(n-1) \geq L + 1}$:}
    %Since $Un \geq (T-1)n + L + 1$, then
    In this case,
    there must be at least two rounds $k$ for which $T < g_k$.  Therefore, by bullet 5 of Theorem~\ref{thm:structure} we have $T \leq g_1$, which further implies by monotonicity that all rounds $k$ satisfy $T \leq g_k$.
    Therefore, $Tn$ of the total gas units are already accounted for,  it remains to allocate the $(U-T)n$ remaining gas units,
    and for each $i$ the additional gas on top of $T$ is between $0$ and $L-T$. 

    By bullet 3, there can be at most one round $i$ such that $g_i \notin \{T,L\}$, and $i$ if exists we have $T < g_i < L$.
    The only monotone sequence $(g_k)$ satisfying these requirements is the one specified in the corollary statement.
    % By bullet 3, there can be at most one round $i$ such that $g_i \notin \{T,L\}$, and $i$ if exists we have $T < g_i < L$.  Furthermore, the monotonicity property implies that $g_i \geq T$ for all $i$. 
    % In other words, $Tn$ of the total gas units are already accounted for, 
    % it remains to allocate the rest which is $(U-T)n$ in total,
    % and for each $i$ the additional gas on top of $T$ is between $0$ and $L-T$. 
    % The only monotone sequence $(g_k)$ satisfying these requirements is the one specified in the corollary statement.
    
\end{proof}

\section{The Two Target Mechanism}
\label{sec:two_target_mechanism}

In this section we consider the two-target mechanism.
Our goal is to find a simply structured optimal strategy as we did for the single-target mechanism in section~\ref{sec:single_target}.
However, as opposed to the single-target setting, here the optimal strategy might not have such a structure.  In particular the optimal strategy might be non-monotonic.
\begin{example}
    Consider an instance with parameters
%    \begin{align*}
%        n &= 4 \\
%        L &= 10 \\
%        U &= 7 \\
%        p_{\min} &= 1 \\
%        \ell_{\mathrm{start}} = h_{\mathrm{start}} &= 0 \\
%        T_\ell &= 1 \\
%        T_h &= 8 \\
%        A_\ell &= 15/8 \\
%        A_h &= 5/8 
%    \end{align*}
    \begin{align*}
        n &= 4,\quad L = 10,\quad  U = 7,\\
        p_{\min} &= 1, \quad \ell_{\mathrm{start}} = h_{\mathrm{start}} = 0,\\
        T_\ell &= 1,\quad A_\ell = 15/8,\\
        T_h &= 8,\quad A_h = 5/8.
    \end{align*}
    A brute-force search over all possible feasible strategies shows that the only optimal strategy is $(10,0,8,10)$, achieving $\approx$ 50.82 revenue, whereas the optimal non-decreasing strategy $(2,8,8,10)$ achieves $\approx$ 51.37. 
\end{example}

The example shows that the optimal strategy in the two-target setting can be very complex in general. 
However, there is some regime of the parameters in which we are in fact able to identify a simply structured strategy, namely in any instance where $L\leq T_\ell + T_h$ (see Appendix~\ref{app:sum_exceeds_limit}). 
%\bbc{TODO:  put the relevant theorem and section in an appendix if there is time, otherwise remove this paragraph.}
Since that regime is not the one we encounter in practice, we instead
% and we therefore do not attempt to characterize it.
% Instead, we 
demonstrate in this section that simply structured strategies give a good approximation to the optimal revenue in general.

To this end, fix as in section~\ref{sec:single_target} the parameters $n$, $U$, $L$, and $p_{\min}$. 
Let us also fix the two-target mechanism parameters $T_\ell,A_\ell,T_h,A_h$
as described in section~\ref{sec:two-target-preliminaries}.
We assume in this section that the initial backlogs are empty, i.e. $\ell_{\mathrm{start}} = h_{\mathrm{start}} = 0$, and that $U \geq T_\ell$.
% \bbc{Can this assumption be removed??}
% \rfc{No, it is needed to argue that the early rounds can be filled up to $T_\ell$ (since gas pays $p_{min}$ then). And it's needed to argue that the high backlog is empty during the first part meaning we can treat it like the one target case.}

\begin{definition}
Let \(\mathcal G_S\subseteq\mathcal G\) be the set of strategies of the form 
\[
\big(\
    \underbrace{T_\ell,\ldots,T_\ell,\;\; T_\ell+e_\ell}_{i_\ell\text{ rounds}},\;\;
    \underbrace{T_h,\ldots,T_h,\;\; T_h+e_h}_{i_h\text{ rounds}},\;\;
    \underbrace{L,\ldots,L}_{i_L\text{ rounds}},
\big),
\]
where $i_\ell+i_h+i_L=n$, $0\leq i_\ell, i_h, i_L\leq n$, $0\leq e_\ell<T_h-T_\ell$ and $0\leq e_h< L-T_h$.
\end{definition}

%Let \(\mathcal G_S\subseteq\mathcal G\) be the set of strategies of the form 
%\[
%\big(\,T_\ell,\ldots,T_\ell,\;\; [g_\ell],\;\; T_h,\ldots,T_h,\;\; [g_h],\;\; L,\ldots,L\,\big),
%\]
%where the bracketed rounds are optional transition rounds with $T_\ell<g_\ell<T_h$ and $T_h<g_h<L$.

In the following, we prove
$$
\min_{G_s\in \mathcal G_S} \mathrm{Rev}(G_S)
\geq
\min_{G\in\mathcal G} \mathrm{Rev}(G)
\geq
c \min_{G_S\in \mathcal G_S} \mathrm{Rev}(G_S),
$$
where \(c<1\) is a constant that depends only on the mechanism parameters. In particular, $c$ is close to 1 for parameters used in practice, meaning that the optimal revenue in $\mathcal{G}_S$ closely approximates that in $\mathcal{G}$.

While the first inequality follows from $\mathcal{G}_S\subseteq\mathcal G$, the second is proven by the following theorem.

\begin{theorem}
\label{thm:approximation}
For every \(G\in\mathcal G\), there exists \(G_S \in\mathcal G_S\) such that
$$
\mathrm{Rev}(G_S) \leq \frac{1}{c}\,\mathrm{Rev}(G),
$$
where
$$
c = \frac{\lambda L}{ \theta + e^{\lambda L - \theta}-1
},
$$

$$
\lambda = \frac1{A_\ell T_\ell}+\frac1{A_h T_h},
\qquad
\theta := \frac{1}{A_\ell}+\frac{1}{A_h}.
$$
\end{theorem}

\begin{proof}
We fix \(G=(g_i)_{i=1}^n\in\mathcal G\), and begin by transforming it into a strategy $\overline{G}$ of the desired form as follows.

Before giving a formal definition, we describe the transformation intuitively.
$\overline{G}$ delays gas usage to later rounds compared to $G$, such that it uses no more than $T_h$ gas each round (except for the last rounds).
Instead it uses up to $T_h$ gas in later rounds where usage was below $T_h$ in $G$, and up to $L$ in the final rounds to accommodate any remaining gas.
See Figure \ref{fig:strategy_transform} for all illustration.
\input{strategy_transformation_figure}

Formally, the transformation is defined as follows. Let
$$
s_i:=\min\{T_h,\;h_i+g_i\}
$$
be the gas amount that is removed from the $T_h$-backlog in round $i$ of $G$, i.e. $h_{i+1}=h_i+g_i-s_i$. Define $\overline{g_i}$ round by round as
$$
\overline{g_i}:=\max\{s_i,\;R_i-(n-i)L\},
\qquad
R_i:=nU-\sum_{k<i}\overline{g_k},
$$
where $R_i$ is the gas that $\overline{G}$ still has to use in rounds $i,\ldots,n$. That is, $\overline{G}$ uses the capped amount $s_i$, unless it means that the remaining gas would not fit into the remaining $n-i$ rounds %even at usage $L$, 
in which case it uses exactly the excess.

$\overline{G}$ is a feasible strategy: Clearly $\overline{g_i}\ge s_i\ge0$.
Moreover, 
we have $\overline{g_i}\le L$ since $s_i \leq T_h \leq L$ and since $R_i\le(n-i+1)L$ --- 
the latter is true for $i=1$ since $R_1=nU\le nL$ and for $i\geq 2$ we indeed have
\[R_i=R_{i-1}-\overline{g_{i-1}}\le(n-i+1)L\] 
where the inequality holds by definition of $\overline{g_{i-1}}$.

Let us examine the structure of $\overline{G}$ in more detail.
Let $l$ be the first round in which the second expression attains the maximum; First we claim that $l$ exists:
if the maximum is not attained by the second expression in the first $n-1$ rounds, i.e.
$\overline{g_k}=s_k$ for all $k<n$, then we have %$R_n=g_n+h_n\ge s_n$. 
\begin{align*}
    R_n &= nU - \sum_{i=1}^{n-1}s_i \\
    &= \sum_{i=1}^n g_n - \sum_{i=1}^{n-1} h_i +g_i - h_{i+1} \\
    &= g_n + h_n \geq s_n
\end{align*}
implying that the maximum is attained by the second expression at round $n$.

For $i<l$ we have $\overline{g_i}=s_i\le T_h$. Moreover $R_{l+1}=(n-l)L$, so inductively $\overline{g_i}=\max\{s_i,L\}=L$ and $R_{i+1}= (n-i)L$ for $i>l$.
In particular, $R_{n+1}=0$, i.e., $\overline{G}$ uses the same total $nU$ as $G$.

Finally, $\overline{G}$ only delays gas compared to $G$:
\begin{equation}
\overline{C_i}:=\sum_{k\le i}\overline{g_k}\;\le\;C_i:=\sum_{k\le i}g_k,
\qquad 0\le i\le n,
\label{eq:cumulative_bounds}
\end{equation}
since $\overline{C_i}=\sum_{k\le i}s_k=C_i-h_{i+1}$ for $i<l$, 
%and $\overline{C_i}=nU-(n-i)L\le C_i$ for $i\ge l$.
whereas for $i\geq l$ we have
\begin{align*}
    \overline{C_i} &= \overline{g_i} + \sum_{k=1}^{i-1} \overline{g_k} \\
    &= R_i - (n-i)L + \sum_{k=1}^{i-1} \overline{g_k} \\
    &= nU - (n-i)L \\
    &\leq C_i
\end{align*}
where the inequality holds since otherwise $G$ would not have enough gas space in the final $n-i$ rounds for the remaining gas budget after fixing the usage in the first $i$ rounds.

To summarize, $\overline{G}$ consists of a first part of rounds with usages $\leq T_h$: either $[1,l-1]$ or $[1,l]$ depending on whether $\overline{g_l}>T_h$. The remaining second part consists of round usages $>T_h$ where all but at most the earliest have precisely usage $L$. This second part is already of the desired form, while the first might not be.

With backlogs initially being empty, the $T_h$-backlog remains empty during the first part, meaning the price behaves according to a mechanism with a sole target $T_\ell$.
That means this part can be transformed 
%to the desired structure 
without increasing revenue using 
%the arguments from the one-target setting,
%Corollary~\ref{cor:final}, 
Theorem~\ref{thm:structure},
where 
$T_h$ is treated as the round gas limit.  Since all steps of this transformation can only decrease the starting price of the second part of the overall strategy (see Observation~\ref{obs:ending_price_can_only_decrease}), 
we end up with a strategy $\widehat{G}$
with $\mathrm{Rev}(\widehat{G}) \leq \mathrm{Rev}(\overline G)$.
Notice that $\widehat{G}$ is of the form
\[
\widehat{G} \;=\; \big(\,\underbrace{0,\ldots,0}_{a\text{ rounds}},\;\; [u],\;\; \underbrace{T_\ell,\ldots,T_\ell}_{b\text{ rounds}},\;\; [v],\;\; \underbrace{T_h,\ldots,T_h}_{c\text{ rounds}},\;\; 
[w],\;\; \underbrace{L,\ldots,L}_{d\text{ rounds}}\,\big),
\]
where $u,v,w$ are potential single transition rounds.  Also note that if $a>0$ or $u<T_\ell$ then $\widehat{G}\notin \mathcal{G}_S$.  However, in that case we can further transform $\widehat{G}$ into a $\mathcal{G}_S$ strategy while not increasing revenue.
\begin{lemma}\label{lem:zero_prefix_transformation}
    There exists $G_{S}\in\mathcal{G}_{S}$ such that $\mathrm{Rev}(G_S) \leq \mathrm{Rev}(\widehat{G})$
\end{lemma}

\begin{proof} Starting from \(\widehat{G}\), repeat the following operation as long as there is a round in the initial below-\(T_\ell\) prefix. Let \(i\) be the rightmost round in this prefix for which \(g_i<T_\ell\). Choose a round \(j>i\) as follows: \[ j= \begin{cases} \min\{k:g_k>T_h\},& \text{if some round satisfies }g_k>T_h,\\[1mm] \min\{k:g_k>T_\ell\},& \text{otherwise}. \end{cases} \] Transfer one unit of gas from round \(j\) to round \(i\): \[ g_i\leftarrow g_i+1, \qquad g_j\leftarrow g_j-1. \] 
%This single rule subsumes the four cases described above. Indeed, 
We argue that every iteration of this operation preserves the structure of $\widehat{G}$.
To this end,
if some round exceeds \(T_h\), then $j$ is the \(w\)-round, if it exists, and otherwise the first \(L\)-round.  
In both cases, the structure is maintained after transferring the single gas unit. 
If no round exceeds \(T_h\), then $j$ is the \(v\)-round, if it exists, and otherwise the first \(T_h\)-round.
Here too, the structure is maintained after transferring the single gas unit.

%Choosing the rightmost deficient recipient and the earliest donor in the relevant region preserves the displayed block form. 
Note that each such transfer does not increase revenue, since the recipient round $i$ satisfies $p_i = p_{\min}$ (because all preceding rounds use less gas than $T_\ell$ and $\ell_{\mathrm{start}} = h_{\mathrm{start}}=0$), implying that the transferred gas unit is charged at most what it was charged before the transfer, prices between rounds $i$ and $j$ (including) do not change, and all prices after round $j$ can only decrease.

Finally, since $U > T_\ell$, we are guaranteed that the below-$T_\ell$ prefix is eliminated before the rounds $j$ as defined above cease to exist.  Consequently, the final strategy uses at least $T_\ell$ gas in every round, i.e., it belongs to $\mathcal{G}_S$ as required.

\end{proof}

What remains is to prove the revenue bound $\mathrm{Rev}(\overline G) \leq \frac1c\,\mathrm{Rev}(G)$.
To do so, we 
%regard gas as continuously divisible and 
rank gas in order of arrival. The ranking of gas in both $G$ and $\overline{G}$ allows the gas to be matched by rank between the two strategies.
The $x$-th gas unit ($0<x\leq g_i$) of round $i$ in $G$ has the cumulative rank $C_{i-1}+x$. 
Let
$$
\tau_i(x):=\min\left\{j:\overline{C_j}\geq C_{i-1}+x \right\}
$$
be the round in which this gas unit is used in \(\overline G\).
The gas unit is charged for at round $i$ with backlogs $\ell_i$ and $h_i$ under $G$, but at round $\tau_i(x)$ with backlogs $\overline{\ell_{\tau_i(x)}}$ and $\overline{h_{\tau_i(x)}}$ under $\overline{G}$. In particular, since $\overline{G}$ only delays gas compared to $G$, $\tau_i(x)\ge i$.

In order to bound the difference between the prices the unit is charged under $G$ and $\overline{G}$, we prove
%the following relationship between these backlogs:
an upper bound on the backlog values $\overline{\ell_{\tau_i(x)}}$, $\overline{h_{\tau_i(x)}}$ for every $i$ and $x$.
\begin{lemma}
\label{lem:g_hat_backlog_bounds}
For all $i$ and $0<x\leq g_i$ we have
\begin{equation}
\begin{aligned}
    \overline{\ell_{\tau_i(x)}} &\leq \ell_i, & \overline{h_{\tau_i(x)}} &\leq h_i & &\text{if } \tau_i(x)=i\\
    \overline{\ell_{\tau_i(x)}} &\leq \ell_i + x-T_\ell, & \overline{h_{\tau_i(x)}} &\leq h_i + x-T_h & &\text{if } \tau_i(x)>i
\end{aligned} %\label{eq:g_hat_backlog_bounds}
\end{equation}
\end{lemma}

\begin{proof}
Denote
$$
z_i:=C_{i-1}-\overline{C_{i-1}}=\sum_{k<i}\big(g_k-\overline{g_k}\big)
$$
denote the amount of gas that $G$ has used before round $i$ but $\overline{G}$ has not yet used. By \eqref{eq:cumulative_bounds} and $\overline{g_k}\ge s_k$,
\begin{equation}
0\le z_i\le\sum_{k<i}(g_k-s_k)=h_i.
\label{eq:z_bounds}
\end{equation}
We also note that $\ell_i\ge h_i$ for all $i$, which follows by induction from $T_\ell<T_h$.
 
We first show that the delayed gas from the first $i$ rounds accounts exactly for the difference in the slow backlogs:
\begin{claim}
\label{clm:l_identity}
For all $i$ we have
\[\ell_i=\overline{\ell_i}+z_i\]
\end{claim}

\begin{proof}
Both sides vanish for $i=1$. Assume that the claim holds for some $i$. Since $z_{i+1}=z_i+g_i-\overline{g_i}$,
$$
\ell_i+g_i-T_\ell
=
\big(\overline{\ell_i}+\overline{g_i}-T_\ell\big)+z_{i+1}.
$$
If $\overline{\ell_i}+\overline{g_i}-T_\ell\ge0$, the right-hand side is nonnegative, and taking $\max\{\,\cdot\,,0\}$ on both sides gives $\ell_{i+1}=\overline{\ell_{i+1}}+z_{i+1}$. Otherwise $\overline{\ell_{i+1}}=0$ and the right-hand side is strictly smaller than $z_{i+1}$. If $z_{i+1}>0$, this would give $\ell_{i+1}<z_{i+1}\le h_{i+1}\le \ell_{i+1}$ by \eqref{eq:z_bounds}, a contradiction. Hence $z_{i+1}=0$, the right-hand side is negative, and $\ell_{i+1}=0=\overline{\ell_{i+1}}+z_{i+1}$.
\end{proof}
 
For the fast backlogs, the delayed gas only yields an upper bound:
\begin{claim}
\label{clm:h_bound}
For all $i$ we have
$$\overline{h_i}\le h_i-z_i$$
\end{claim}

\begin{proof}
Denote
$$r_i = \sum_{k<i}\big(\overline{g_k}-s_k\big)= h_i - z_i$$
We have $r_i\ge0$ since $\overline{g_k}\ge s_k$, $r_{i+1}=r_i+\overline{g_i}-s_i$, and $\overline{h_1}=0=r_1$. If the claim holds for some $i$, then, using $s_i\le T_h$ and $r_{i+1}\ge0$,
$$
\overline{h_{i+1}}
=\max\{\overline{h_i}+\overline{g_i}-T_h,\,0\}
\le\max\{r_i+\overline{g_i}-s_i,\,0\}
=r_{i+1}.
$$
\end{proof}

Now fix a gas unit $0<x\le g_i$ of round $i$ and write $j:=\tau_i(x)$. If $j=i$, the first line 
%of \eqref{eq:g_hat_backlog_bounds}
of the lemma statement
follows directly from Claims~\ref{clm:l_identity}, \ref{clm:h_bound} and $z_i\ge0$.
 
If $j>i$, then by the definition of $\tau_i(x)$,
$$
\overline{C_k}<C_{i-1}+x\le C_k
\qquad\text{for } k=i,\ldots,j-1.
$$
Hence $z_{k+1}>0$, and thus $h_{k+1}>0$ by \eqref{eq:z_bounds}.  This in turn implies that $s_k=T_h$ and therefore $\overline{g_k}\ge s_k=T_h>T_\ell$ for $k=i,\ldots,j-1$. In these rounds neither backlog of $\overline{G}$ is truncated at zero, i.e., $\overline{\ell_{k+1}}=\overline{\ell_k}+\overline{g_k}-T_\ell$ and $\overline{h_{k+1}}=\overline{h_k}+\overline{g_k}-T_h$. %Summing over $k$ and 
% Writing $d:=j-i\ge1$, we obtain 
%from $\overline{C_{j-1}}<C_{i-1}+a$ and \eqref{eq:l_identity}
% $$
% \overline{\ell_j}
% =\overline{\ell_i}+\big(\overline{C_{j-1}}-\overline{C_{i-1}}\big)-dT_\ell
% <\overline{\ell_i}+z_i+x-dT_\ell
% =\ell_i+x-dT_\ell
% \le \ell_i+x-T_\ell,
% $$
From here we get
\begin{align*}
    \overline{\ell_j} &= \overline{\ell_j} + \left(\sum_{k=i}^{j} \overline{\ell_k}\right) -
    \left(\sum_{k=i}^{j} \overline{\ell_k}\right) \\
    &= \overline{\ell_j} +\overline{\ell_i} + \left(\sum_{k=i}^{j-1} \overline{\ell_k} +\overline{g_k} - T_\ell \right)  -
    \left(\sum_{k=i}^{j} \overline{\ell_k}\right) \\
    &= \overline{\ell_i} + \left(\overline{C_{j-1}}-\overline{C_{i-1}}\right) -(j-i)T_\ell \\
    &< \overline{\ell_i} + (z_i+x) -(j-i)T_\ell \\
    &=\ell_i +x - (j-i)T_\ell \\
    &\leq \ell_i + x - T_\ell
\end{align*}
where the first inequality holds since $\overline{C_{j-1}}<C_{i-1}+a$ and by definition of $z_i$, and the last equality holds by Claim~\ref{clm:l_identity}.
Analogous reasoning for the fast backlogs (using Claim~\ref{clm:h_bound}) obtains
% $$
% \overline{h_j}
% <\overline{h_i}+z_i+x-dT_h
% \le h_i+x-dT_h
% \le h_i+x-T_h.
% $$
\[
\overline{h_j} < h_i +x - T_h.
\]
This proves %\eqref{eq:g_hat_backlog_bounds}.
the lemma.
\end{proof}

Now we can use the backlog bounds in %\eqref{eq:g_hat_backlog_bounds} 
Lemma~\ref{lem:g_hat_backlog_bounds}
to bound the price at which the $x$-th gas unit in round $i$ of $G$ is processed in $\overline{G}$:
\begin{align}
    \overline{p_{\tau_i(x)}} &= p_{\min} \exp\left( \frac{\overline{\ell_{\tau_i(x)}}}{A_\ell T_\ell}+\frac{\overline{h_{\tau_i(x)}}}{A_h T_h} \right) \nonumber \\
    &\leq  p_{\min} \exp\left( \frac{\ell_i}{A_\ell T_\ell}+\frac{h_i}{A_h T_h} + \max(0, \lambda x-\theta) \right) \nonumber \\
    &= p_i e^{\max(0, \lambda x-\theta)}
\label{eq:g_hat_price_bound}
\end{align}
Note that already at this point we could use the fact that $x\leq L$ to derive the bound 
\[\mathrm{Rev}(\overline G) \leq  e^{\lambda L-\theta}\cdot \mathrm{Rev}(G).\]
%\bbc{Note that the exponent here is exactly $\frac{L-T_l}{A_l T_l} + \frac{L-T_h}{A_h T_h}$  we usually constrain to be $\ln(1.2)$!!!  This would correspond to $c\approx 83\%$}
%\rfc{better bound below}
%
However, we can get a better bound with the following approach.
Consider the average price that gas originally arriving in round $i$ of $G$ pays in $\overline{G}$:
$$
\frac{1}{g_i} \int_0^{g_i} \overline{p_{\tau_i(x)}} \,dx.
$$
Using \eqref{eq:g_hat_price_bound}, this can be upper bounded by
$$
 p_i \frac{1}{g_i} \int_0^{g_i} e^{\max(0, \lambda x-\theta)}\,dx.
$$
To show that the new average price is at most $\frac{1}{c}$ times the original price, it remains to prove that
$$
\frac{1}{g_i} \int_0^{g_i} e^{\max(0,\lambda x-\theta)}\,dx \leq \frac{1}{c}.
$$
The left-hand side (LHS) is the average of the integrand over $[0, g_i]$. Since the integrand is non-decreasing, the LHS is non-decreasing in $g_i$.
Because $g_i\leq L$ the largest possible value of LHS is attained for $g_i = L$. For $g_i = L$ the integral evaluates to
$$
\dfrac{\theta}{\lambda} + \dfrac{e^{\lambda L - \theta}-1}{\lambda},
$$
which proves the claim.

\end{proof}

Note that $e^{\lambda L - \theta}$ is the factor by which the price increases when the largest possible amount of gas $L$ is used, since the exponent equals $\frac{L-T_l}{A_l T_l} + \frac{L-T_h}{A_h T_h}$. Therefore this expression is the maximum factor $1+\tau$ by which the price can increase from one round to the next.
In practice, this value is often fixed to ensure good user experience, as described in Section~\ref{sec:introduction}.
%\footnote{A fixed maximum price increase allows users - or their wallets - to set a reasonable maximum gas price such that transactions get reliably included on-chain.}

Assume the maximal increase is set to $\tau>0$, i.e. $\lambda L - \theta = \ln (1+\tau)$ and consider how $1/c$ depends on $\tau$:
$$
\frac{1}{c} = \frac{ \theta + e^{\lambda L - \theta}-1}{\lambda L} = 1+\frac{\tau-\ln (1+\tau)}{\lambda L}
$$
Since $\lambda L\geq \ln(1+\tau)$, we observe
$$
\frac{1}{c} \leq \frac{\tau}{\ln(1+\tau)} < 1 + \frac{\tau}{2}.
$$
So for any gas price mechanism where the maximum increase is set to $\tau$, the previous theorem provides a $\frac{\tau}{2}$-approximation of the optimal revenue.

For our running example (Example~\ref{running_example}) of practical parameters ($L = 192~\mathrm{Mgas/s}$, $T_\ell = 10~\mathrm{Mgas/s}$, $A_\ell = 86400$, $T_h = 60~\mathrm{Mgas/s}$, $A_h = 12$), the theorem's constant is $c \approx 93.7\%$.
So it guarantees that the overall optimal revenue is within $6.3\%$ fo the revenue of the optimal simply-structured strategy.
%Note that this is just an upper bound on the difference. For the parameters in question, the actual difference is much smaller. \rfc{maybe mention a number?}
%\bbc{Not sure I would go into this, as you'd presumably need to explain how you found out the revenue of the optimal strategy.}

\section{Computing the Optimal Simply-Structured Strategy Revenue}
\label{sec:computing_optimal_revenue}

The previous sections prove that the revenue of the optimal strategy is equal to or close to the revenue of the optimal simply-structured strategy.
In the following, we discuss how to compute the optimal simply-structured strategy revenue in practice.

\subsection{Single Target}
\label{sec:single_target_computing_optimal_revenue}
Assuming that the backlog is initially empty (equivalently, the initial price satisfies $p_1 = p_{\min}$) and that $U\geq T$, Corollary \ref{cor:final} shows that there is an optimal strategy $G^*$ of the form
\[
    G^* \;=\; \big(\,\underbrace{T,\ldots,T,\; T+r}_{n-q\text{ rounds}},\; \underbrace{L,\ldots,L}_{q\text{ rounds}}\,\big)
\]
where
\[
    q = \left\lfloor \frac{n(U-T)}{L-T} \right\rfloor, \qquad r = (U-T)n - (L-T)q.
\]

With this strategy, the price remains at $p_{\min}$ during the first $n-q$ rounds. In the next round, the price is $p_{min}\exp\left(\frac{r}{AT}\right)$, and increases by a factor of $\exp\left(\frac{L-T}{AT}\right)$ in each following round.
Hence, adding up the round revenues leads to a geometric sum which simplifies to
\[
    \mathrm{Rev}(G^*) = p_{min}\left( (n-q)T+r + L\exp\left(\frac{r}{AT}\right)\frac{\exp\left(\frac{q(L-T)}{AT}\right)-1}{\exp\left(\frac{L-T}{AT}\right)-1} \right).
\]

\subsection{Two Targets}
Assume both backlogs are initially empty and $U\geq T_\ell$. In Section \ref{sec:two_target_mechanism} we proved that the smallest possible revenue is close to the smallest revenue of a strategy of the form
\[
\big(\
    \underbrace{T_\ell,\ldots,T_\ell,\;\; T_\ell+e_\ell}_{i_\ell\text{ rounds}},\;\;
    \underbrace{T_h,\ldots,T_h,\;\; T_h+e_h}_{i_h\text{ rounds}},\;\;
    \underbrace{L,\ldots,L}_{i_L\text{ rounds}},
\big),
\]
where $i_\ell+i_h+i_L=n$, $0\leq i_\ell, i_h, i_L\leq n$, $0\leq e_\ell<T_h-T_\ell$ and $0\leq e_h< L-T_h$.

The revenue of such a strategy can be written as a closed form expression in these parameters similarly to the single target strategy.

Note that due to the total-gas constraint we have
\[
    e_\ell + i_h(T_h-T_\ell) + e_h + i_L(L-T_\ell) = n(U-T_\ell).
\]
In particular, given values for $e_h$ and $i_L$, the remaining variables $e_\ell$, $i_h$ and $i_\ell$ are uniquely determined: the total-gas constraint can be solved for $e_\ell$ and $i_h$, since $0\leq e_\ell<T_h-T_\ell$, and then $i_\ell$ can be inferred from $i_\ell+i_h+i_L=n$.
Thus the minimal revenue attainable by strategies of the form above can be found by trying the $n \times (L - T_h)$ possible combinations for $i_L$ and $e_h$.

In practice however, this can be sped up significantly:

\paragraph{Reducing the search space for $\bm{e_h}$}
The search space of $e_h$ can be reduced by considering a grid of possible values -- for instance, multiples of $b = 0.001\min(T_h, L-T_h)$.
The loss in precision can be bounded as follows:

Consider an arbitrary simply-structured strategy $G\in \mathcal{G_S}$.
Transform it to $G'\in \mathcal{G_S}$ by rounding $e_h$ up to the next multiple of $b$, and reducing $e_\ell$ by the same amount (if $e_\ell$ is smaller than the rounded amount, subtract the remaining amount from round $i_\ell+1$).
The latter can only decreases the revenue of $G'$ compared to $G$.
The increase of $e_h$ changes the revenue in that round by a factor of at most 1.001 since $b\leq 0.001T_h$.
Moreover, since both backlogs in all $L$-rounds increase by at most $b$ when compared to $G$, this implies that the prices in all $L$-rounds increase by a factor of at most
$$\exp\left({b\left(\frac{1}{A_lT_l}+\frac{1}{A_hT_h}\right)}\right) \leq \exp\left( \frac{b}{L-T_h} \left( \frac{L-T_l}{A_l T_l} + \frac{L-T_h}{A_h T_h} \right) \right) = (1+\tau)^{ \frac{b}{L-T_h}} \leq (1+\tau)^{0.001}.$$
In particular, $(1+\tau)^{0.001} < 1.001$ for $\tau$ in practice ($\tau < 1.7$). 
Together, this means $\mathrm{Rev}(G') \leq 1.001 \mathrm{Rev}(G)$.
So the minimum revenue among the set of rounded strategies is at most $0.1\%$ larger than the minimum of $\mathcal{G_S}$.
Hence trying the $1000 \frac{L-T_h}{\min(T_h, L-T_h)}$ multiples of $b$ for $e_h$ yields an $0.1\%$ approximation of the actual minimum.

\paragraph{Reducing the search space for $\bm{i_L}$}
Because $A_h < A_\ell$ (and in practice, typically, even $A_h\ll A_\ell$) usage above $T_h$ leads to a (significantly) stronger price increase in the following rounds than usage below $T_h$.
In particular, when $i_L$ is large, the gas price can quickly increase during the rounds with usage $L$. 
So for configurations with large $i_L$, the cost of the block of $L$-rounds alone can be high enough to rule out the optimality of the configuration.
This can be used to compute a simple upper bound on $i_L$ as follows:

First, calculate the revenue $\overline{R}$ of the strategy with $i_L=0$ and $e_h=0$ as a baseline (using the method described in section~\ref{sec:single_target_computing_optimal_revenue}).
Any configuration in which the revenue of the $L$-rounds alone is larger than this can be disregarded. 
The revenue obtained from the $L$-rounds alone, assuming conservatively that both backlogs are 0 when the $L$-rounds begin, is given by

\[
    R_L(i_L) := p_{min}L\frac{(1+\tau)^{i_L}-1}{\tau},\qquad\text{where } 1+\tau = \exp\left( \frac{L-T_l}{A_l T_l} + \frac{L-T_h}{A_h T_h} \right).
\]
Then $R_L(i_L)\leq \overline{R}$ implies
\[
    i_L \leq \frac{\log\left(1+ \frac{\tau}{L} \frac{\overline{R}}{p_{min}} \right)}{\log (1+\tau)}.
\]
For our running example, this means $i_L \leq 58$, which massively reduces the number of values to be checked.

\section{Setting Pricing Parameters}
\label{sec:setting_parameters}

We now discuss how the results of the previous sections can be applied in practice to choose the parameters of the two target mechanism.
In particular, they provide a method for choosing the low target $T_\ell$ and its adjustment rate $A_\ell$ to satisfy the economic security guaranty as described in Section~\ref{sec:introduction}.

We take the gas limit $L$ as given and assume the minimum price $p_{\min}$, the high target $T_h$, and its adjustment window $A_h$ to be determined by other design requirements. For instance, $T_h$ and $A_h$ can be chosen to control gas usage on the chain over a short period of time (seconds to minutes) to ensure that nodes can keep up with processing.
This can be done by setting $T_h$ and $A_h$ directly or by choosing  $T_h$ and the maximum per-round price increase $\tau$.
Then $A_h$ can be inferred once all other parameters are set from $\frac{L-T_\ell}{A_\ell T_\ell} + \frac{L-T_h}{A_h T_h} = \ln (1+\tau)$.

We now want to set a long-term constraint, i.e., $T_\ell$ and $A_\ell$, to limit the average gas usage over a longer period of time (days to weeks). In particular, this is commonly done to limit the amount of state growth.
Suppose that state growth constraints dictate that the average gas usage should not exceed $U$ over a certain period of time ($n$ seconds).
To enforce this, we demand that the adversary's budget necessary to incur a higher average usage than $U$ on the chain is at least $B$ -- and want to choose $T_\ell$ and $A_\ell$ accordingly.

For any choice of $T_\ell$ and $A_\ell$, the previous sections provide an efficient method of calculating the cost $\mathrm{Rev}(T_\ell, A_\ell)$ of an (adversarially) optimal gas usage, for which we require that $\mathrm{Rev}(T_\ell, A_\ell) \geq B$.
For given $T$ in the relevant range  $T_\ell \in [0, U]$, we can use binary search to find the appropriate $A_\ell$ that solves the budget constraint equation.
This produces a $(T_\ell, A_\ell)$-curve which summarizes the available design tradeoff.
Picking a larger $T_\ell$ allows for a higher long-term average chain throughput.
On the other hand, a larger $T_\ell$ demands $A_\ell$ to be smaller, which makes the mechanism produce more aggressive price increases.

Figure \ref{fig:parameter_curve_budget} shows the $(T_\ell, A_\ell)$-curve for our running example (see Example~\ref{running_example}): $L = 192M, T_h = 60M, \tau = 20\%, p_{\min}=0.02\ \mathrm{gwei}, n=2{,}592{,}000~\mathrm{s}$ (30 days). The plot shows the curve for a few choices of safe fee revenue amounts $B$, namely $4K$, $40K$ and $400K$ ETH, which roughly translate to $\$10M$, $\$100M$ and $\$1B$, respectively.

\begin{figure}
    \centering
    \includegraphics[width=0.8\linewidth]{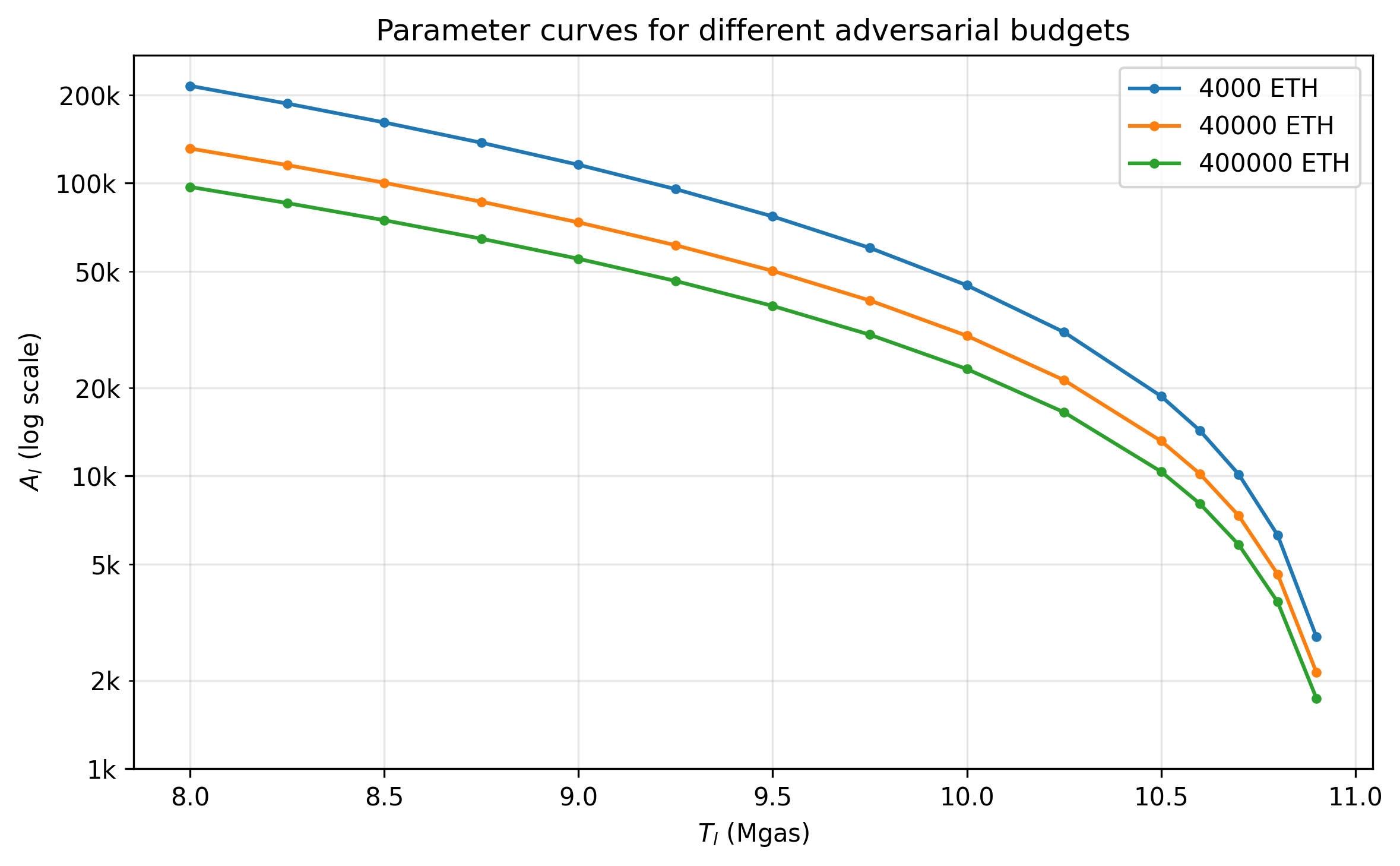}
    \caption{Possible $(T_l, A_l)$ parameter choices for $U=11 Mgas$, $T_h=60Mgas$, $\tau = 20\%$.}
    \label{fig:parameter_curve_budget}
\end{figure}

\section{Conclusion}

In this work, we presented a framework for selecting economically secure parameters for exponential EIP-1559-style pricing mechanisms. The parameters are chosen to guarantee a lower bound on the total fees needed to sustain a given average gas usage over a fixed time window. This lets the mechanism designer make it prohibitively expensive to violate a usage constraint, without making any assumptions about gas demand.

For the single-target mechanism, we characterized an exact revenue-minimizing gas distribution. For the two-target mechanism with initially empty backlogs, we identified a simply structured distribution whose revenue provably approximates the minimum, with an approximation guarantee controlled by the maximum per-round relative price increase. These structural results make the relevant fee bounds efficiently computable and provide a basis for selecting gas targets and adjustment rates to meet a prescribed economic security threshold.

Our analysis makes explicit the tradeoff between higher throughput targets and more gradual price adjustments. Moving the target farther below the desired usage rate creates slack that permits a less aggressive adjustment rate while still preserving the economic security guarantee. In the two-target setting, this supports calibrating long-term constraints, such as state growth, alongside short-term requirements for guaranteeing nodes can keep up with the chain. 

The resulting framework provides a theoretical basis for parameter selection in existing mechanisms and for the design of future variants.

\paragraph{AI Disclaimer}
The single-target analysis was developed entirely by the authors.
The two-target approximation result (Theorem~\ref{thm:approximation}) was developed with assistance from OpenAI's \emph{GPT-5.6 Sol} model:
The authors proposed the structure of the approximate strategies, and ChatGPT derived an initial bound.
The authors subsequently identified several improvements that strengthened the bound, which ChatGPT incorporated into the proof.
Finally, the authors refined and rewrote the proof for clarity.
The authors take full responsibility for the content and correctness of the manuscript.

%The two-target approximation result (Theorem~\ref{thm:approximation}) was developed through conversations with ChatGPT: the authors posed the question of whether an approximation was possible, and ChatGPT provided an initial proof draft, which the authors subsequently refined and rewrote for clarity. The authors take full responsibility for the content and correctness of the manuscript.

\bibliographystyle{plainurl}
\bibliography{references}

\appendix

\section{The Two Target Mechanism when the Sum of Targets Exceeds the Limit}
\label{app:sum_exceeds_limit}

Fix the parameters $n$, $U$, $L$, and $p_{\min}$, and let us also fix the two-target mechanism parameters $T_\ell,A_\ell,T_h,A_h$
as described in section~\ref{sec:two-target-preliminaries}.
In this section we further assume that $L \leq T_\ell + T_h$.
Our goal in this section is to prove
\begin{restatable}{theorem}{sumExceedsLimit}
\label{thm:sum_exceeds_limit}
If $L \leq T_\ell + T_h$, then there exists an optimal strategy $G$ satisfying:
\begin{enumerate}
    \item $G$ is monotone non-decreasing.
    \item There is at most one round $i$ for which $0 < g_i < T_\ell$. 
    \item There is at most one round $i$ for which $T_\ell < g_i < T_h$.
    \item There is at most one round $i$ for which $T_h < g_i < L$.
\end{enumerate}   
\end{restatable}

We start with The following generalization of Lemma~\ref{lem:rev G - rev G' cancel out}.

\begin{lemma}
\label{lem:rev G - rev G' cancel out generalized}
    Let $(G,G')$ be a $j$-shifted strategy pair. Suppose that for each $d \in \{\ell,h\}$ the following two conditions hold:
    \begin{enumerate}
        \item $g_j \geq T_d$ or $g_{j+1} \leq T_d$
        \item $g'_j \geq T_d$ or $g'_{j+1} \leq T_d$
    \end{enumerate}
    Then $p_{j+2} = p'_{j+2}$, and therefore the difference in revenue between $G$ and $G'$ is determined only by rounds $j$ and $j+1$, i.e.
    \[
    \mathrm{Rev}(G) - \mathrm{Rev}(G') = 
    p_j g_j + p_{j+1} g_{j+1} - \left(p_j g'_j + p'_{j+1} g'_{j+1}\right)
    \]
\end{lemma}

\begin{proof}
    From applying Lemma~\ref{lem:rev G - rev G' cancel out} to both price trajectories 
    %$r$ and $q$,
    $p_{\ell,\,i} := p_{\min}\exp\left(\frac{\ell_i}{A_\ell T_\ell}\right), p_{h,\,i} := \exp\left(\frac{h_i}{A_h T_h}\right)$
    it follows that 
    %$r_{j+2} = r'_{j+2}$
    $p_{\ell,\, {j+2}} = p'_{\ell,\, {j+2}}$
    and 
    %$q_{j+2} = q'_{j+2}$
    $p_{h,\, {j+2}} = p'_{h,\, {j+2}}$
    which proves $p_{j+2} = p'_{j+2}$.
    
\end{proof}

Given 
%$r_i,q_i,g_i$
$p_{\ell,\, i}, p_{h,\, i}, g_i$
we define the relative price increase function
% \[
% M(r_i,q_i,g_i) := \frac{p_{i+1}}{p_i} = 
% \frac{r_{i+1}q_{i+1}}{r_i q_i}
% \]
\[
M(p_{\ell,\, i}, p_{h,\, i}, g_i) := \frac{p_{i+1}}{p_i} = 
\frac{p_{\ell,\, i+1}\cdot p_{h, \,i+1}}{p_{\ell, \,i}\cdot p_{h, \,i}}
\]
% \rfc{TODO: add some explanation to this definition since the use of non-arguments is confusing}
where we note that $p_{\ell,\, i+1}$ and $p_{h, \,i+1}$ are determined from $p_{\ell,\, i}, p_{h, \,i}$ and $g_i$.
We will sometimes abuse notation and write $M(g_i)$ - which assumes the starting price state
%$r_i$ and $q_i$
$p_{\ell, i}$ and $p_{h, i}$
are given from the context.
\begin{lemma}
\label{lem:M_properties}
    Fix pricing state 
    %$r_i, q_i$.
    $p_{\ell, \,i}$ and $p_{h, \,i}$.
    Then we have:
    \begin{enumerate}
        \item $M$ is non-decreasing in $g_i$
        \item $M(T_\ell) \leq 1$.
        \item The function $g \mapsto \frac{M(g) - 1}{g}$ is non-decreasing.
    \end{enumerate}
\end{lemma}

\begin{proof}
    The first bullet is trivial. The second bullet holds since a gas usage of $g_i = T_\ell$ leaves the $\ell$ backlog the same, whereas the $h$ backlog can only decrease (with strict inequality if $p_{h,\,i} > 1)$.

    As for bullet 3, if we denote $\phi(g) = \ln(M(g))$, then clearly $\phi(g)$ is convex, implying in turn that $F(g):= M(g) - 1$ is also convex.
    Now fix \(0<g_1<g_2\). Since
$
g_1
=
\frac{g_1}{g_2}g_2
+
\left(1-\frac{g_1}{g_2}\right)0,
$
convexity of \(F\) gives
\[
F(g_1)
\leq
\frac{g_1}{g_2}F(g_2)
+
\left(1-\frac{g_1}{g_2}\right)F(0).
\]
Using \(F(0) = M(0)-1 \leq M(T_\ell)- 1 \leq 0 \), we obtain
$
F(g_1)
\leq
\frac{g_1}{g_2}F(g_2),
$ which implies
\[
\frac{M(g_1)-1}{g_1}
= \frac{F(g_1)}{g_1}
\leq
\frac{F(g_2)}{g_2}
=
\frac{M(g_2)-1}{g_2},
\]
as desired.
    
\end{proof}

\begin{proposition}
\label{prop:g_j>T_r, g_j + g_{j+1} - T_r leq T_q}
    Let $G$ be a strategy such that $g_j > T_\ell$ and $g_j + g_{j+1} - T_\ell \leq T_h$, for some round $j$. Consider the strategy $G'$ obtained from $G$ by setting $g'_j = T_\ell$, $g'_{j+1} = g_j +g_{j+1} - T_\ell$.  Then $\mathrm{Rev}(G') \leq \mathrm{Rev}(G)$. 
\end{proposition}

\begin{proof}
    We observe that the assumptions of the proposition imply the conditions of Lemma~\ref{lem:rev G - rev G' cancel out generalized}.  
    We therefore have
    \begin{align*}
        \mathrm{Rev}(G) - \mathrm{Rev}(G') &= 
    p_j g_j + p_{j+1} g_{j+1} - \left(
    p_j T_\ell + p'_{j+1} (g_j + g_{j+1} - T_\ell)
    \right) \\
    \end{align*}
    After dividing by $p_j$, the right-hand side becomes
    \begin{align*}
        g_j + M(g_j) g_{j+1} - T_\ell - M(T_\ell)(g_j + g_{j+1} -T_\ell) & = (g_j - T_\ell) (1 - M(T_\ell)) + g_{j+1}(M(g_j) - M(T_\ell))
    \end{align*}
    Both summands in the right-hand side are non-negative by Lemma~\ref{lem:M_properties}, and we are done.

\end{proof}

The following proposition corresponds to the first bullet of Theorem~\ref{thm:sum_exceeds_limit}.  We separate it from the proof of the theorem for readability.

\begin{proposition}
\label{prop:structure_2_targets}
If $L \leq T_\ell + T_h$, then there exists an optimal strategy $G$ which is monotone non-decreasing.
\end{proposition}

\begin{proof}
    Let $G$ be an optimal strategy, and suppose there is a round $j$ for which $g_{j+1} < g_j$.  We will show the existence of a dominating strategy $G'$ whose obtained revenue is at most that of $G$ (see definition of the domination relation in Section~\ref{sec:single_target}).
    We split to cases.

    Assume for the first case that $g_j$ and $g_{j+1}$ (weakly) lie in the same regions, i.e. 
    % lie on the same side of $T_h$, for both $h \in\{r,q\}$.
    for both $d \in \{\ell,h\}$ we have either $g_j,g_{j+1} \leq T_d$ or $g_j,g_{j+1}\geq T_d$.
    In this case let $G'$ be the strategy obtained from $G$ by swapping $g_j$ with $g_{j+1}$, i.e.
    \[
    g'_j = g_{j+1}\ ,\ g'_{j+1} = g_j 
    \]
    We observe that the assumptions of Lemma~\ref{lem:rev G - rev G' cancel out generalized} hold, and therefore we have
    \begin{align*}
        \mathrm{Rev}(G) - \mathrm{Rev}(G') &= 
    p_j g_j + p_{j+1} g_{j+1} - \left(p_j g'_j + p'_{j+1} g'_{j+1}\right) \\
    &= p_j g_j +p_{j+1} g_{j+1} -(p_j g_{j+1} + p'_{j+1} g_j)
    \end{align*}
    Notice that if $g_{j+1} = 0$ then we end up with the expression $g_j(p_j - p'_{j+1})$ which is non-negative since in this case we clearly have $p'_{j+1} \leq p_j$.  We can thus assume that $g_j > g_{j+1}> 0$, which allows us to divide the resulting expressions by $p_j\cdot g_j \cdot g_{j+1}$.  The right-hand side becomes
    \begin{align*}
        \frac{1}{g_{j+1}} + \frac{p_{j+1}}{p_j g_j} - \frac{1}{g_j} - \frac{p'_{j+1}}{p_j g_{j+1}} &=
        \frac{1}{g_{j+1}} + \frac{M(g_j)}{g_j} - \frac{1}{g_j} - \frac{M(g_{j+1})}{g_{j+1}} \\
        &= \frac{M(g_j) - 1}{g_j} - \frac{M(g_{j+1}) - 1}{g_{j+1}}
    \end{align*}
    and this is non-negative by bullet 3 of Lemma~\ref{lem:M_properties}.  This concludes the first case.

    For the second case, assume that $g_{j+1} < T_\ell < g_j$ and $g_{j} + g_{j+1} - T_\ell \leq T_h$.
    Note that the assumptions of Proposition~\ref{prop:g_j>T_r, g_j + g_{j+1} - T_r leq T_q} hold in this case, and so there exists a strategy $G'$ as desired.
    % In this case, define
    % \[
    % g'_j = T_r \ ,\ g'_{j+1} = g_j + g_{j+1} - T_r
    % \]
    % and we again observe that the condition of Lemma~\ref{lem:rev G - rev G' cancel out generalized} holds.  We therefore have
    % \begin{align*}
    %     \mathrm{Rev}(G) - \mathrm{Rev}(G') &= 
    % p_j g_j + p_{j+1} g_{j+1} - \left(
    % p_j T_r + p'_{j+1} (g_j + g_{j+1} - T_r)
    % \right) \\
    % \end{align*}
    % After dividing by $p_j$, the right-hand side becomes
    % \begin{align*}
    %     g_j + M(g_j) g_{j+1} - T_r - M(T_r)(g_j + g_{j+1} -T_r) & = (g_j - T_r) (1 - M(T_r)) + g_{j+1}(M(g_j) - M(T_r))
    % \end{align*}
    % Both summands in the right-hand side are non-negative by Lemma~\ref{lem:M_properties}, and we are done.
    % \bbc{Note:  the argument above holds also for the case $T_r < g_j < g_{j+1}$ and $g_{j} + g_{j+1} - T_r \leq T_q$}

    Now suppose both previous cases do not hold.  If $g_{j+1} < T_\ell$ then we must have $g_j > T_\ell$ as otherwise we are back in the first case.  
    This further implies that
    if $g_{j+1} < T_\ell$ then $g_j + g_{j+1} - T_\ell > T_h$ as otherwise (together with $g_j > T_\ell$) we are in the second case, and in particular $g_j > T_h$.  
    
    If, on the other hand, we have $T_\ell \leq g_{j+1}$, then we must also have $g_j > T_h$ (as otherwise we have $T_\ell \leq g_{j+1} < g_j \leq T_h$ and we are back in the first case) which in particular also implies $g_j + g_{j+1} - T_\ell > T_h$.  
    
    To conclude, when both previous cases do not hold, we must have both $g_j > T_h$ and $g_j + g_{j+1} -T_\ell > T_h$.
    We define 
    \[
    g'_j = g_j + g_{j+1} -T_h\ ,\ g'_{j+1} =T_h
    \]
    Note that $g'_j > T_\ell$ and that the condition of Lemma~\ref{lem:rev G - rev G' cancel out generalized} holds.  We therefore have
    
    \begin{align}
        \frac{\mathrm{Rev}(G) - \mathrm{Rev}(G')}{p_j} &= \frac{1}{p_j}\left[ 
    p_j g_j + p_{j+1} g_{j+1} - \left(
    p_j g'_j + p'_{j+1} g'_{j+1}
    \right)
    \right] \nonumber \\
    &=g_j +M(g_j)g_{j+1} - (g_j+g_{j+1} - T_h) - M(g'_j)T_h \nonumber \\
    &= g_{j+1} (M(g_j) - 1) - T_h (M(g'_j) - 1) \label{eq:case_3_inequality}
    \end{align}
We need to show that the above is non-negative.
First, note that $g_j > T_h$ implies that $M(g_j) > 1$.  In particular, we can assume that $M(g'_j) > 1$, as otherwise the claim is trivial.  Therefore, we assume equivalently that $0 < \ln (M(g'_j)) < \ln(M(g_j))$.
We claim that it suffices to show that

\begin{equation}
\label{eq:case_3_main_ineq}
T_h \ln (M(g'_j)) \leq g_{j+1} \ln (M(g_j)).
\end{equation}
To see that this is indeed enough, note that \eqref{eq:case_3_main_ineq} implies

\begin{align*}
    T_h (M(g'_j) - 1) &= T_h \ln (M(g'_j)) \frac{M(g'_j) - 1}{\ln (M(g'_j))}  \leq g_{j+1} \ln (M(g_j)) \frac{M(g_j) - 1}{\ln (M(g_j))} = g_{j+1}(M(g_j) - 1)   
\end{align*}
where the inequality holds by equation \eqref{eq:case_3_main_ineq}, the fact that the function $z \mapsto \frac{e^z - 1}{z}$ is increasing, and the fact that $g'_j \leq g_j$.
%and $\ln (M(g_j)) > 0$ (since $g_j > T_q$).

We now prove that inequality \eqref{eq:case_3_main_ineq} holds.  To this end, 
recall that $l_j, h_j$ are the gas backlog values at round $j$, corresponding to the low and high targets respectively.
We first observe that
% \begin{align*}
% \ln M(g_j) &= \max\left\{R_j + \frac{g_j - T_r}{A_r T_r}, 0\right\} - R_j + \max\left\{Q_j + \frac{g_j - T_q}{A_q T_q} , 0\right\} - Q_j \\
% &= \frac{g_j - T_r}{A_r T_r} + 
% \frac{g_j - T_q}{A_q T_q}\\
% \ln M(g'_j) &= \max\left\{R_j + \frac{g'_j - T_r}{A_r T_r}, 0\right\} - R_j + \max\left\{Q_j + \frac{g'_j - T_q}{A_q T_q} , 0\right\} - Q_j \\
% &= \frac{g'_j - T_r}{A_r T_r} + 
% \max\left\{ \frac{g'_j - T_q}{A_q T_q} , -Q_j\right\} 
% \end{align*}

\begin{align*}
\ln M(g_j) &= \max\left\{\frac{\ell_j + g_j - T_\ell}{A_\ell T_\ell}, 0\right\} - \frac{l_j}{A_\ell T_\ell} + \max\left\{\frac{h_j + g_j - T_h}{A_h T_h} , 0\right\} - \frac{h_j}{A_h T_h } \\
&= \frac{g_j - T_\ell}{A_\ell T_\ell} + 
\frac{g_j - T_h}{A_h T_h}\\
\ln M(g'_j) &= \max\left\{\frac{\ell_j + g'_j - T_\ell}{A_\ell T_\ell}, 0\right\} - \frac{l_j}{A_\ell T_\ell} + \max\left\{\frac{h_j + g'_j - T_h}{A_h T_h} , 0\right\} - \frac{h_j}{A_h T_h } \\
&= \frac{g'_j - T_\ell}{A_\ell T_\ell} + 
\max\left\{ \frac{g'_j - T_h}{A_h T_h} , -\frac{h_j}{A_h T_h }\right\} 
\end{align*}
We establish inequality \eqref{eq:case_3_main_ineq} separately for the $\ell$ and the $h$ summands.
\paragraph{$\ell$-summand} We need to show that 
\[
\frac{T_h(g'_j - T_\ell)}{A_\ell T_\ell} \leq \frac{g_{j+1}(g_j - T_\ell)}{A_\ell T_\ell}.
\]  
After removing the common denominator, the left-hand side becomes
\begin{align*}
    T_h(g'_j - T_\ell) &= T_h(g_j + g_{j+1} - T_h - T_\ell) \\
    &= T_h(g_j - T_\ell) - T_h(T_h - g_{j+1})
\end{align*}
Subtracting this from the right-hand side, we get
\begin{align*}
    g_{j+1}(g_j - T_\ell) - T_h(g_j - T_\ell) + T_h(T_h - g_{j+1} ) &= 
    (T_h - g_{j+1})(T_\ell - g_j) + T_h(T_h - g_{j+1} ) \\
    &= (T_h - g_{j+1})(T_\ell + T_h - g_j)
\end{align*}
and this is indeed non-negative as desired, since $g_{j+1} < T_h$ and $g_j \leq L \leq T_\ell + T_h$. 
%\bbc{This is the first place in the proof that we use $L \leq T_r + T_q$}.

\paragraph{$h$-summand}
We need to show that
\[
T_h\cdot \max\left\{ \frac{g'_j - T_h}{A_h T_h} , -\frac{h_j}{A_h T_h }\right\}  \leq g_{j+1}\cdot \frac{g_j - T_h}{A_h T_h}.
\]
Since the right-hand side is non-negative and $-h_j$ is non-positive, it is enough to show that
\[
T_h (g'_j - T_h) \leq g_{j+1}(g_j - T_h).
\]
% It is enough to show that $T_q (g'_j - T_q) \leq g_{j+1}(g_j - T_q)$ and that 
% %$-T_q Q_j \leq - g_{j+1} Q_j$.
% $-T_q Q_j \leq  g_{j+1}(g_j - T_q)$.
% The latter holds trivially since $T_q < g_j$, so we prove the former.  
%
Subtracting the left-hand side from the right-hand side, we get

\begin{align*}
    g_{j+1}(g_j - T_h) - T_h (g'_j - T_h) &= g_{j+1}(g_j - T_h) - T_h (g_j + g_{j+1} - T_h - T_h) \\
    &= g_{j+1}g_j - g_{j+1}T_h +2T_h^2 - T_h g_j - T_h g_{j+1} \\
    &= g_{j+1}g_j - 2g_{j+1}T_h +2T_h^2 - T_h g_j \\
    &=(T_h - g_{j+1})(2T_h - g_j)
\end{align*}
and this is non-negative as desired since $g_{j+1} \leq T_h$ and $g_j \leq L \leq T_\ell + T_h < 2T_h$. 
%\bbc{This is the second place we used that assumption}

\end{proof}

We are now ready to prove Theorem~\ref{thm:sum_exceeds_limit}.  For clarity we restate it here.

\sumExceedsLimit*

% \begin{theorem}
% There is an optimal strategy $G$ satisfying:
% \begin{enumerate}
%     \item $G$ is monotone non-decreasing.
%     \item There is at most one round $j$ for which $0 < g_j < T_r$. 
%     \item There is at most one round $i$ for which $T_r < g_i < T_q$.
%     \item There is at most one round $k$ for which $T_q < g_k < L$.
% \end{enumerate}   
% \end{theorem}

\begin{proof}
    Similar to the proof of Theorem~\ref{thm:structure}, it is enough to show that if any of the bullets does not hold but the preceding bullets do hold for any strategy $G$, then there exists a strategy $G'$ that dominates $G$ and whose obtained revenue is at most that which is obtained by $G$.

    We already proved this for bullet 1 (see Proposition~\ref{prop:structure_2_targets}).  Now let $G$ satisfy bullet 1 but not 2, i.e. $G$ is a monotone non-decreasing strategy and there is a round $j$ such that $0 < g_j \leq g_{j+1} < T_\ell$.  
    Consider the strategy $G'$ defined as $G$ only that $g'_j = g_j - 1$, $g'_{j+1} = g_{j+1} +1$,
    and we note that the proof of Proposition~\ref{prop:g_j_leq_T_g_j+1+leq_T-1} carries over verbatim to our setting.

    Suppose now that $G$ satisfies bullets 1,2 but not 3.  Then there is some round $j$ such that $T_\ell < g_j \leq g_{j+1} < T_h$.  If $g_j+g_{j+1} - T_\ell \leq T_h$ then we finish by applying Proposition~\ref{prop:g_j>T_r, g_j + g_{j+1} - T_r leq T_q}.  
    
    Otherwise, let $G'$ be the strategy obtained from $G$ by setting $g'_j = g_j + g_{j+1} - T_h$ and $g'_{j+1} = T_h$.  We observe that the conditions of Lemma~\ref{lem:rev G - rev G' cancel out generalized} hold, and so we have
    \begin{align*}
        \frac{\mathrm{Rev}(G) - \mathrm{Rev}(G')}{p_j} &=
        g_j + M(g_j) g_{j+1} - (g'_j + M(g'_j)g'_{j+1}) \\
        &= g_j + M(g_j) g_{j+1} - ( g_j + g_{j+1} -T_h + M(g'_j)T_h) \\
        &= g_{j+1} (M(g_j) - 1) - T_h (M(g'_j) - 1)
    \end{align*}
    and this last expression is non-negative since 
    \begin{align*}
        T_h (M(g'_j) - 1) &= T_h \cdot g'_j\frac{(M(g'_j) - 1)}{g'_j} \\
        &\leq  g_{j+1}\cdot g_j \frac{(M(g'_j) - 1)}{g'_j} \\
        &\leq g_{j+1}\cdot g_j \frac{(M(g_j) - 1)}{g_j} \\
        &= g_{j+1}(M(g_j) - 1),
    \end{align*}
    where the first inequality holds by basic calculus arguments since the parabolic function $f(x) := x(g_j + g_{j+1} - x)$ satisfies
    \[
    T_h \cdot g'_j = f(g'_j) < f(g_j) = g_j \cdot g_{j+1},
    \]
    and the second inequality holds by bullet 3 of Lemma~\ref{lem:M_properties} (since $g'_j < g_j$).

    Suppose now that $G$ satisfies bullets 1,2,3 but not 4, implying that there is a round $j$ such that $T_h < g_j \leq g_{j+1} < L$.  If $g_j + g_{j+1} - T_h > L$, then we can proceed very similarly to the argument in the previous paragraph.  Explicitly, we set $g'_j = g_j + g_{j+1} - L$ and $g'_{j+1} = L$. We observe that the conditions of Lemma~\ref{lem:rev G - rev G' cancel out generalized} hold, which implies that
    \begin{align*}
        \frac{\mathrm{Rev}(G) - \mathrm{Rev}(G')}{p_j} &=
        g_j + M(g_j) g_{j+1} - (g'_j + M(g'_j)g'_{j+1}) \\
        &= g_j + M(g_j) g_{j+1} - ( g_j + g_{j+1} -L + M(g'_j)L) \\
        &= g_{j+1} (M(g_j) - 1) - L (M(g'_j) - 1)
    \end{align*}
    The last expression is non-negative since
    \begin{align*}
        L (M(g'_j) - 1) &= L \cdot g'_j\frac{(M(g'_j) - 1)}{g'_j} \\
        &\leq  g_{j+1}\cdot g_j \frac{(M(g'_j) - 1)}{g'_j} \\
        &\leq g_{j+1}\cdot g_j \frac{(M(g_j) - 1)}{g_j} \\
        &= g_{j+1}(M(g_j) - 1)
    \end{align*}
    and we are done.

    We are left with the case $g_j + g_{j+1} - T_h \leq L$.  We define $g'_j = T_h$ and $g'_{j+1} = g_j + g_{j+1} - T_h$ and observe that the conditions of Lemma~\ref{lem:rev G - rev G' cancel out generalized} hold.  From here we have
    \begin{align}
    \label{eq:rev_diff_last_case}
        \frac{\mathrm{Rev}(G) - \mathrm{Rev}(G')}{p_j} &=
        g_j + M(g_j) g_{j+1} - (g'_j + M(g'_j)g'_{j+1}) \nonumber \\
        &= g_j + M(g_j) g_{j+1} - (T_h + M(T_h)(g_j + g_{j+1} - T_h)) \nonumber \\
        &= g_{j+1}(M(g_j) - M(T_h)) - (g_j - T_h)(M(T_h) - 1) \nonumber \\
        &\geq T_h M(T_h)\left(\frac{M(g_j)}{M(T_h)} - 1 \right) - (g_j - T_h)(M(T_h) - 1)
    \end{align}
    where the inequality holds since $g_{j+1} \geq T_h$ (and $M(g_j) - M(T_h) \geq 0$ by Lemma~\ref{lem:M_properties}).

    It remains to show that the right-hand side of~\eqref{eq:rev_diff_last_case} is non-negative.
    To that end we establish a relation between $M(g_j)$ and $M(T_h)$.  When $T_h$ gas is used, the $\ell$ component of the price increases whereas the $h$ component stays the same.  We therefore have
    \begin{equation}
    \label{eq:M_T_q}
    M(T_h) = \frac{p'_{j+1}}{p_j} = \exp\left(\frac{T_h -T_\ell}{A_\ell T_\ell}\right)
    \end{equation}
    Furthermore, since $g_j >T_h$, then when $g_j$ gas is used both components increase.  Combining with \eqref{eq:M_T_q}, we get
    \begin{align}
    \label{eq:M_g_j_over_M_T_q}
        M(g_j) &= \exp\left(\frac{g_j - T_\ell}{A_\ell T_\ell} + \frac{g_j - T_h}{A_h T_h}\right) \nonumber \\
        &= \exp\left(\frac{T_h-T_\ell}{A_\ell T_\ell}\right)
        \exp\left((g_j - T_h)\left(\frac{1}{A_\ell T_\ell} + \frac{1}{A_h T_h}\right)\right) \nonumber \\
        &=M(T_h) \exp\left((g_j - T_h)\left(\frac{1}{A_\ell T_\ell} + \frac{1}{A_h T_h}\right)\right) 
    \end{align}
    At this point we introduce notation.  Let $\Delta := g_j - T_h$ and $\lambda := \frac{1}{A_\ell T_\ell} + \frac{1}{A_h T_h} $, and so \eqref{eq:M_g_j_over_M_T_q} becomes
    \[
    \frac{M(g_j)}{M(T_h)} = \exp(\Delta \lambda)
    \]
    Substituting this into \eqref{eq:rev_diff_last_case}, we get
    \begin{align*}
        \frac{\mathrm{Rev}(G) - \mathrm{Rev}(G')}{p_j} &= 
        T_h M(T_h)\left(\exp\left(\Delta \lambda\right)  - 1 \right) -\Delta(M(T_h) - 1) \\
        &\geq T_h M(T_h)\cdot \Delta \lambda -\Delta(M(T_h) - 1) \\
        &= \Delta (T_h M(T_h)\lambda -M(T_h) + 1)
    \end{align*}
    where the inequality holds by the identity $\exp(x) - 1 \geq  x$.  It is therefore left to prove that $T_h M(T_h)\lambda -M(T_h) + 1 \geq 0 $, or equivalently 
    \[
    T_h \lambda \geq 1 - 1/M(T_h)
    \]
    To this end we have
    \begin{align*}
    1 - 1/M(T_h)    &= 1 - \exp\left(-\frac{T_h -T_\ell}{A_\ell T_\ell}\right) \\
                    &\leq \frac{T_h -T_\ell}{A_\ell T_\ell} \\
                    &\leq T_h \left(\frac{1}{A_\ell T_\ell} + \frac{1}{A_h T_h}\right) \\
                    &= T_h \lambda
    \end{align*}
    where the first inequality again uses the identity $1 -\exp(x) \leq -x $.  This concludes the proof.
    
\end{proof}

\end{document}

%% file: settings/packages.tex
\usepackage{amsmath,amsthm,graphicx,url}
\usepackage{amssymb}
\usepackage[margin=1in]{geometry}
\usepackage{bm}
\usepackage{mathtools}
\usepackage[normalem]{ulem}
\usepackage{thmtools}

\usepackage{tikz}
\usetikzlibrary{arrows.meta,positioning}

\usepackage{footmisc}

\usepackage{cancel}

\usepackage{soul}

\usepackage{booktabs} % For formal tables

\usepackage[ruled, noline]{algorithm2e} % For algorithms

\SetKwBlock{Initialize}{Initialize:}{}
\SetKwFor{From}{from}{do}{}
\SetKwFor{ForAllInParallel}{for all}{in parallel do}{}

\SetCommentSty{mycommfont}

\usepackage{xurl}
\usepackage[colorlinks=true, urlcolor=blue, linkcolor=red]{hyperref}
\usepackage{cleveref}

\usepackage{thmtools} 
\usepackage{thm-restate}

\usepackage[numbers]{natbib}

\SetAlFnt{\small}
\SetAlCapFnt{\small}
\SetAlCapNameFnt{\small}
\SetAlCapHSkip{0pt}
\IncMargin{-\parindent}

\makeatletter
\g@addto@macro\@floatboxreset\centering
\makeatother

%% file: settings/macros.tex
\usepackage[]{color-edits}% suppress
\addauthor{EF}{blue}
\addauthor{RF}{red}
\addauthor{BB}{cyan}

%% file: strategy_transformation_figure.tex
\begin{figure}[htb]
\centering

\begin{tikzpicture}[
  x=0.65cm,
  y=0.36cm,
  every node/.style={font=\small},
]

\def\Tq{6}
\def\L{10}
\def\barhalf{0.42}   % Bar width = 0.88; gap = 0.12
\def\panelshift{9}  % Horizontal offset of the second panel

% ------------------------------------------------------------
% Original strategy G: left panel
% ------------------------------------------------------------

\begin{scope}
  \node at (4.5,12)
    {Original strategy $G$};

  \foreach \i/\h in {
    1/4,
    2/8,
    3/5,
    4/3,
    5/8,
    6/8,
    7/10,
    8/4
  }{
    \filldraw[fill=gray!20,draw=black]
      (\i-\barhalf,0)
      rectangle
      (\i+\barhalf,\h);

    \node[font=\scriptsize,anchor=south]
      at (\i,\h+0.12) {$\h$};
  }
\end{scope}

% ------------------------------------------------------------
% Transformed strategy: right panel
% ------------------------------------------------------------

\begin{scope}[shift={(\panelshift,0)}]
  \node at (4.5,12)
    {Transformed strategy $\overline{G}$};

  \foreach \i/\h in {
    1/4,
    2/6,
    3/6,
    4/4,
    5/6,
    6/6,
    7/8,
    8/10
  }{
    \filldraw[fill=gray!20,draw=black]
      (\i-\barhalf,0)
      rectangle
      (\i+\barhalf,\h);

    \node[font=\scriptsize,anchor=south]
      at (\i,\h+0.12) {$\h$};
  }
\end{scope}

% ------------------------------------------------------------
% Baselines and reference lines for both panels
% ------------------------------------------------------------

\foreach \s in {0,\panelshift}{
  \begin{scope}[shift={(\s,0)}]
    % Horizontal line at zero
    \draw[black]
      (0.45,0) -- (8.55,0);

    % High target
    \draw[dashed,black]
      (0.45,\Tq) -- (8.55,\Tq);

    % Gas limit
    \draw[densely dotted,black]
      (0.45,\L) -- (8.55,\L);

    % Reference labels on the left panel only
    \ifnum\s=0\relax
      \node[anchor=east,font=\scriptsize]
        at (0.35,0) {$0$};
    
      \node[anchor=east,font=\scriptsize]
        at (0.35,\Tq) {$T_q=6$};

      \node[anchor=east,font=\scriptsize]
        at (0.35,\L) {$L=10$};
    \fi
  \end{scope}
}

\end{tikzpicture}

\caption{Illustration of the strategy transformation}
\label{fig:strategy_transform}
\end{figure}